\documentclass{article}

\usepackage{fullpage}

\usepackage[utf8]{inputenc} 
\usepackage[T1]{fontenc}    

\usepackage{url}            
\usepackage{booktabs}       
\usepackage{amsfonts}       
\usepackage{nicefrac}       
\usepackage{microtype}      
\usepackage{xcolor}         

\usepackage{hyperref}\hypersetup{                
                    colorlinks=true,                
                    breaklinks=true,                
                    urlcolor= black,                
                    linkcolor= blue,              
                    bookmarksopen=false,
                    citecolor=black,
                    linkcolor=black,
                    filecolor=black,
                    citecolor=blue,
                    linkbordercolor=blue
}       

\usepackage[utf8]{inputenc}

\usepackage{float}
\usepackage{amsmath,amssymb}
\usepackage[utf8]{inputenc} 
\usepackage[T1]{fontenc}    
\usepackage{hyperref}       
\usepackage{url}            
\usepackage{booktabs}       
\usepackage{amsfonts,amsthm}       
\usepackage{nicefrac}       
\usepackage{microtype}      
\usepackage{lipsum}
\usepackage{graphicx}
\usepackage{bbm}

\def\countsketch{\text{\texttt{CountSketch}}}
\def\sketch{\text{\texttt{Sketch}}}

\usepackage[algo2e, ruled, vlined]{algorithm2e}
\usepackage{xcolor} 

\SetCommentSty{mycommfont}

\newtheorem{thm}{Theorem}[section]
\newtheorem{theorem}{Theorem}[section]

\newtheorem{lemma}[thm]{Lemma}

\newtheorem{remark}[thm]{Remark}

\newtheorem{definition}[thm]{Definition}

\def\E{\textsf{E}}
\def\Var{\textsf{Var}}
\def\Pr{\textsf{Pr}}
\def\supp{\texttt{supp}}

\providecommand{\ignore}[1]{}
\newcommand{\tail}{\mathsf{tail}}

\title{Tricking the Hashing Trick: A Tight Lower Bound on the Robustness of CountSketch to Adaptive Inputs}

\author{Edith Cohen\thanks{Google Research and Tel Aviv University. \texttt{edith@cohenwang.com}.} \and Jelani Nelson\thanks{UC Berkeley and Google Research. \texttt{minilek@alum.mit.edu}.} \and Tam\'{a}s Sarl\'{o}s\thanks{Google Research. \texttt{stamas@google.com}.} \and Uri Stemmer\thanks{Tel Aviv University and Google Research. \texttt{u@uri.co.il}. Partially supported by the Israel Science Foundation (grant 1871/19)
and by Len Blavatnik and the Blavatnik Family foundation.}}
\date{}

\begin{document}

\maketitle

\begin{abstract}
  CountSketch and Feature Hashing  (the "hashing trick")
 are popular randomized dimensionality reduction
methods that support recovery of $\ell_2$-heavy hitters (keys $i$ where $v_i^2 > \epsilon \|\boldsymbol{v}\|_2^2$) and approximate inner products. When the inputs are {\em not adaptive} (do not depend on prior outputs), classic estimators applied to a sketch of size $O(\ell/\epsilon)$ are accurate for a number of queries that is exponential in $\ell$.  When inputs are adaptive, however, an adversarial input can be constructed after $O(\ell)$ queries with the classic estimator and the best known robust estimator only supports $\tilde{O}(\ell^2)$ queries.
  In this work we show that this quadratic dependence is in a sense inherent: We design an attack
that after $O(\ell^2)$ queries produces an adversarial input vector whose sketch is
highly biased. Our attack uses "natural" non-adaptive inputs (only the final adversarial input is chosen adaptively) and universally applies  with any correct estimator, including one that is unknown to the attacker. In that, we expose inherent vulnerability of this fundamental method.
\end{abstract}

\section{Introduction}
\label{sec:introduction}

\texttt{CountSketch}~\cite{CharikarCFC:2002} and its variant {\em feature hashing}~\cite{MoodyD:NIPS1989,WeinbergerDLSA:ICML2009} are 
immensely popular dimensionality reduction methods that map 
input vectors in $\mathbb{R}^n$ to their sketches in $\mathbb{R}^d$ (where $d\ll n$). 
The methods have many applications in machine learning and data analysis and often are used as components in large models or pipelines~\cite{WeinbergerDLSA:ICML2009,pmlr-v5-shi09a,PhamPagh:KDD2013,ChenWTWC:ICML2015,KDD-2016-ChenWTWC,AmiraliSLDSB:ICML2018,SpringKMA:ICML2019,AhleKKPVWZ:SODA2020,CPW:NEURIPS2020}. 

The mapping is specified by 
{\em internal randomness} $\rho\sim \mathcal{D}$ that determines a set of
$d=\ell\cdot b$ linear measurements vectors $(\boldsymbol{\mu}^{(j,k)})_{j\in [\ell],k\in[b]}$ in $\mathbb{R}^n$.
The sketch $\sketch_\rho(\boldsymbol{v})$ of a vector $\boldsymbol{v}\in \mathbb{R}^n$  is the matrix of $d$ linear measurements 
\begin{equation} \label{sketchdef:eq}
\sketch_{\boldsymbol{\rho}}(\boldsymbol{v}) := \left(\langle \boldsymbol{\mu}^{(j,k)},\boldsymbol{v}  \rangle\right)_{j\in[\ell], k\in [b]}  .
\end{equation}
The salient properties of \texttt{CountSketch} are that (when setting $b=O(1/\epsilon)$ and $\ell = O(\log n)$) the $\ell_2$-heavy hitters of an input $\boldsymbol{v}$, that is, keys $i$ with $v_i^2 > \epsilon \|\boldsymbol{v}\|_2^2$, can be recovered from $\sketch_\rho(\boldsymbol{v})$ and that the inner product of two vectors $\boldsymbol{v}$, $\boldsymbol{u}$ can be approximated from their respective sketches $\sketch_\rho(\boldsymbol{v})$, $\sketch_\rho(\boldsymbol{u})$. This recovery is performed by applying an appropriate {\em estimator} to the sketch, for example, the median estimator~\cite{CharikarCFC:2002} provides estimates on values of keys and supports heavy hitters recovery. But recovery can also be implicit, for example, when the sketch is used as a compression module in a Neural Network~\cite{ChenWTWC:ICML2015}, the recovery of features is learned.

  
Randomized data structures and algorithms, \texttt{CountSketch} included, are typically analysed under an assumption that the input sequence is generated in a way that does not depend on prior outputs and on the sketch randomness $\rho$. This assumption, however, does not always hold, for example, when there is an intention to construct an adversarial input or when the system has a feedback between inputs and outputs~\cite{SpringKMA:ICML2019,FetchSGD:ICML2020}.

The adaptive setting, where inputs may depend on prior outputs, is more challenging to analyse and 
there is 
growing interest in quantifying performance and in designing methods that are robust to adaptive inputs. Works in this vein span machine learning~\cite{szegedy2013intriguing,goodfellow2014explaining,athalye2018synthesizing,papernot2017practical}, adaptive data analysis~\cite{Freedman:1983,Ioannidis:2005,FreedmanParadox:2009,HardtUllman:FOCS2014,DworkFHPRR15}, dynamic graph algorithms~\cite{ShiloachEven:JACM1981,AhnGM:SODA2012,gawrychowskiMW:ICALP2020,GutenbergPW:SODA2020,Wajc:STOC2020, BKMNSS22}, and sketching and streaming algorithms~\cite{MironovNS:STOC2008,AhnGM:SODA2012,HardtW:STOC2013,BenEliezerJWY21,HassidimKMMS20,WoodruffZ21,AttiasCSS21,BEO21,GuptaJNRSW:NeurIPS2021,CLNSSS:ICML2022}. 
Robustness to adaptive inputs can trivially be achieved by using a fresh data structure for each query, or more finely, for each time the output changes.  Hence, $\ell$ independent replicas of a non-robust data structures suffice for supporting $\ell$ adaptive queries. A powerful connection between adaptive robustness and differential privacy \cite{DworkFHPRR15} and utilizing the workhorse of DP composition, yielded essentially a wrapper around $\ell$ independent replicas of a non-robust data structure that supports a quadratic number $\tilde{O}(\ell^2)$ of adaptive queries (or changes to the output) \cite{HassidimKMMS20,GuptaJNRSW:NeurIPS2021,BKMNSS22}. 
For the problem of recovering heavy-hitters from $\countsketch$, the "wrapper" method supports $\approx \epsilon \ell^2$ adaptive queries.  The current state of the art~\cite{CLNSSS:ICML2022}  is a robust estimator that works with a variant of $\countsketch$ and supports $\tilde{O}(\ell^2)$ adaptive queries.

Lower bounds on the performance of algorithms to adaptive inputs are obtained by designing an {\em attack}, a sequence of input vectors, that yields a constructed input that is adversarial to the internal randomness $\rho$. Tight lower bounds on the robustness of statistical queries were established by \cite{HardtUllman:FOCS2014,SteinkeUllman:COLT2015}, who designed an attack with a number of queries that is quadratic in the sample size, which matches the known upper bounds \cite{DworkFHPRR15}.  Their construction was based on fingerprinting codes \cite{BonehShaw:Crypto1995}. A downside of these constructions is that the inputs used in the attack are not "natural" and hence unlikely to shed some understanding on practical vulnerability in the presence of feedback. Hardt and Woodruff \cite{HardtW:STOC2013} provided an impossibility result for 
the task of estimating the norm of the input within a constant factor from (general) linear sketches.
Their construction works with arbitrary correct estimators and produces an adversarial distribution over inputs where the sketch measurements are "far" from their expectations.  The attack size, however, has a large polynomial dependence on the sketch size and is far from the respective upper bound.
Ben-Eliezer et al \cite{BenEliezerJWY21} present an attack on the AMS sketch~\cite{ams99} for the task of approximating the $\ell_2$-norm of the input vector. The attack is tailored to a simplified estimator that is linear in the set of linear measurements (whereas the "classic" estimator uses a median of measurements and is not linear). Their attack is efficient in that the number of queries is of the order of the sketch size, rendering the estimator non-robust. It also has an advantage of using "natural" inputs.
More recently, \cite{CLNSSS:ICML2022} presented attacks that are tailored to  specific estimators for $\countsketch$, including an attack of size $O(\ell)$ on the classic median estimator and an attack of size $O(\ell^2)$ on their proposed robust estimator.

\subsection*{Contribution}

Existing works proposed attacks of size that is far from the corresponding known upper bounds or are tailored to a particular estimator.  Specifically for $\countsketch$, there is an upper bound of $O(\ell^2)$ but it is not even known whether there exist estimators that support a super-quadratic number of adaptive inputs. This question is of particular importance because $\countsketch$ and its variants are the only known efficient sketching method that allow recovery of $\ell_2$-heavy hitters and approximating $\ell_2$ norms and inner products. Moreover, their form as linear measurements is particularly suitable for efficient implementations and integration as components in larger pipelines.  Finally, a recent lower bound precludes hope for an efficient deterministic (and hence fully robust) sketch \cite{KamathPW:CCC2021}, so it is likely that the vulnerabilities of $\countsketch$ are inherent to $\ell_2$-heavy hitter (and approximate $\ell_2$-norm and inner product) recovery from a small sketch.

We construct a {\em universal} attack on \texttt{CountSketch}, that applies against any unknown, potentially non-linear, possibly state maintaining, estimator.  We only require that the estimator is correct.
Our attack uses $O(\ell^2)$ queries, matching the $\tilde{O}(\ell^2)$ robust estimator upper bound \cite{CLNSSS:ICML2022}. Moreover, it suffices for the purpose of the attack that the estimator only reports a set of candidate heavy keys without their approximate values (we only require that heavy hitters are reported with very high probability and $0$ value keys are reported with very small probability). Our attack also applies against a correct inner-product estimator (that distinguishes between $\langle \boldsymbol{v},\boldsymbol{u}\rangle=0$ (reported with very small probability) and  $\langle \boldsymbol{v},\boldsymbol{u}\rangle^2 \geq \epsilon \|\boldsymbol{v}\|_2^2 \|\boldsymbol{u}\|_2^2$ (reported with high probability.))  Additionally, we apply our attack to correct $\ell_2$-norm estimators applied to an AMS sketch~\cite{ams99} and obtain that an attack of size $O(\ell^2)$ suffices to construct an adversarial input. The AMS sketch can be viewed as a $\countsketch$ with $b=1$ and is similar to the Johnson Lindenstrauss transform \cite{JLlemma:1984}.

The product of our attack (with high probability) is an {\em adversarial input} $\boldsymbol{v}$ on which the
measurement values of $\sketch_\rho(\boldsymbol{v})$ are very biased with respect to their distribution when $\rho\sim\mathcal{D}$. Specifically, the design of $\countsketch$ results in  
linear measurements that are unbiased for any input $\boldsymbol{v}$ under the sketch distribution $\rho\sim \cal{D}$: For each key $i$ and measurement vectors $\boldsymbol{\mu}$ with $\mu_i \not = 0$ it holds that 
$\E_{\rho}[\langle \boldsymbol{v},\boldsymbol{\mu }\rangle/ \mu_i- v_i]=0$ but the corresponding expected values for our adversarial  $\boldsymbol{v}$ in $\sketch_\rho(\boldsymbol{v})$ are large 
($\geq B \epsilon \|\boldsymbol{v}\|_2^2$ for a desired $B>1$).
This "bias" means that the known standard (and robust) estimators for heavy hitters and inner products would fail on this adversarial input. And generally estimators that satisfy the usual design goal of being correct on any input with high probability over the distribution of $\sketch_\rho(\boldsymbol{v})$ may not be  correct on the adversarial inputs. 
We note however that our result does not preclude the existence of specialized estimators that are correct on our adversarial inputs. We only show that a construction of an input that is adversarial to $\rho$ is possible with any correct estimator. 

Finally, our attacks uses "natural" inputs that constitute of a heavy key and random noise. The final adversarial input is a linear combination of the noise components according to the heavy hitter reports and is the only one that depends on prior outputs. The simplicity of this attack suggests "practical" vulnerability of this fundamental sketching technique.

\medskip
{\bf Technique}
Our attacks construct an adversarial input with respect to key $h$.  The high level structure is to generate "random tails," $(\boldsymbol{z}^{(t)})_{t\in [r]}$, which are vectors with small random entries. Ideally, we would like to determine for each $\boldsymbol{z}^{(t)}$ whether it is biased up or down with respect to key $h$.  Roughly, considering the set $T_h$ of $\ell$ measurement vectors with $\mu_h \in \{-1,1\}$, determine the sign $s^{*(t)}$ of $\frac{1}{\ell} \sum_{\boldsymbol{\mu} \in T_h} \langle \boldsymbol{\mu},\boldsymbol{z}^{(t)} \rangle \cdot \mu_h$.  
If we had that, the linear combination
$\boldsymbol{z}^{*(A)}= \sum_{t\in[r]} s^{*(t)} \boldsymbol{z}^{(t)}$ (with large enough $r$) is 
an adversarial input. The intuition why this helps is that the bias accumulates linearly with the number of tails $r$ whereas the standard deviation (essentially the $\ell_2$ norm), considering randomness of the selection of tails, increases proportionally to $\sqrt{r}$. 
The attack strategy is then to design query vectors of the form
$v^{(t)}_h \boldsymbol{e}_h + \boldsymbol{z}^{(t)}$
so 
that from 
whether or not $h$ is reported as a heavy hitter candidate
we obtain $s^{(t)}$ that correlates with $s^{*(t)}$. A
higher correlation $\E[s^{*(t)} s^{(t)}]$ yields more effective attacks:
With $\E[s^{*(t)} s^{(t)}] = \Omega(1)$ we get attacks of size $r=O(\ell)$ and with $\E[s^{*(t)} s^{(t)}]=\Omega(1/\sqrt{\ell})$ we get attack of size  $r=O(\ell^2)$. We show that we can obtain $s^{(t)}$ with $\E[s^{*(t)} s^{(t)}]=\Omega(1/\sqrt{\ell})$ (thus matching the upper bound) against any arbitrary and adaptable estimator as long as it is correct. The difficulty is that such an estimator can be non monotone in the value $v_h$ of the heavy key and change between queries. 
Our approach is to select the value $v_h$ of the heavy key in the input vectors  uniformly at random from an interval that covers the "uncertainty region" between values that must be reported as heavy to values that are hard to distinguish from $0$ and hence should not be reported. We then observe that {\em in expectation} over the random choice, any correct estimator must have a slight reporting advantage with larger $v_h$.  Finally, we show that any particular sketch content is equally likely with symmetrically larger $v_h$ with a tail that is biased down or a smaller $v_h$ with a tail that is biased up.  This translates to a slight reporting advantage for $h$ as a candidate heavy hitter when the tail is "biased up."  Therefore by taking $s^{(t)}=1$ when $h$ is reported and $s^{(t)}=-1$ otherwise we have the desired correlation.

\medskip
{\bf Remark:} We note that our use of "noise vectors" that have sparsity larger than sketch size is necessary to break robustness.  This because with high probability, from a $\countsketch$ with parameters $b=O(1/\epsilon)$ and $\ell = O(\log n)$, we can fully recover {\em all} vectors $\boldsymbol{v}\in\mathbb{R}^n$ with at most $O(1/\epsilon)$ non-zeros. Therefore, $\countsketch$ is fully robust to sufficiently sparse inputs, but turns out to be less robust on inputs where a set of 
$O(1/\epsilon)$ heavy entries is augmented with noise.

\medskip
{\bf Related work:} In terms of techniques, our work is most related to \cite{CLNSSS:ICML2022} in that the structure of our attack vectors is similar to those used in \cite{CLNSSS:ICML2022} to construct a tailored attack on the classic estimator. The generalization however to a "universal" attack that is effective against arbitrary and unknown estimators was delicate and required multiple new ideas.  Our contribution is also related and in a sense complementary to \cite{HardtW:STOC2013} that designed attack on linear sketches that applies with any correct estimator for (approximate) norms. Their attack is much less efficient in that its size is a higher degree polynomial and it uses dependent (adaptive) inputs (whereas with our attack only the final adversarial input depends on prior outputs). The product of their attack are constructed vectors that are in the (approximate) null space of the sketching matrix. These "noise" vectors can have large norms but are "invisible" in the sketch. When such "noise" is added to an input with a signal (say a heavy hitter), the "signal" is suppressed (entry no longer heavy) but can still be recovered from the sketch. Our attack fails the sketch matrix in a complementary way: we construct "noise" vectors that do not involve a signal (a heavy entry) but the sketch mimics a presence of that particular signal.

\medskip
{\bf Overview:}
Our attack is described in Section~\ref{sec:attack_structure} and is
analysed in subsequent sections.
The query vectors used in the attack have a simple sketch structure (described in Section~\ref{sketchdist:sec}) and simplify the estimation task (as described in Section~\ref{sec:estimators}). In Section~\ref{simplifiedminest:sec} we cast the attack on heavy hitters estimator as one on a {\em mean estimator} and carry out the remaining analysis in this simplified context. Section~\ref{minestproofs:sec} includes detailed proofs. In Section~\ref{AMS:sec} we apply our attack technique to $\ell_2$ norm estimators for the AMS sketch~\cite{ams99}. 

\section{Preliminaries}
\label{sec:preliminaries}

We use boldface notation for vectors $\boldsymbol{v}$,  non boldface
for scalars $v$, $\langle \boldsymbol{v},\boldsymbol{u} \rangle =\sum_i v_i u_i$ for
inner product,  and $v\cdot u$ for scalar product. For a vector $\boldsymbol{v}\in \mathbb{R}^n$ we refer to $i\in[n]$ as a {\em key} and $v_i$ as the value of the $i$th key (entry) and denote by $\overline{v}= \frac{1}{n} \sum_{i=1}^n v_i$ the mean value.
For exposition clarity, we use $\approx$ to mean "within a small relative error."
We denote by $\mathcal{N}(v,\sigma^2)$ the normal distribution with mean $v$ and variance $\sigma^2$ and by $\boldsymbol{u} \sim \mathcal{N}_\ell(v,\sigma^2)$ a vector in $\mathbb{R}^\ell$ with entries that are i.i.d.\ $\mathcal{N}(v,\sigma^2)$. The probability density function of $\mathcal{N}_\ell(v,\sigma^2)$ is
\begin{equation} \label{pdfnormal:eq}
f_v(\boldsymbol{u}) = 
\prod_{i\in[\ell]} \frac{1}{\sigma\sqrt{2\pi}} e^{-\frac{1}{2} \left(\frac{u_i-v}{\sigma}\right)^2}\ .
\end{equation}

\begin{definition} (heavy hitter) \label{HH:def}
For $\epsilon>0$,  and a vector $\boldsymbol{v}\in \mathbb{R}^n$, key $i\in [n]$ is an $\ell_2$-$\epsilon$-{\em heavy hitter} if $v_i^2 \geq \epsilon  \|\boldsymbol{v} \|_2^2$.
\end{definition}
Clearly, there can be at most $1/\epsilon$ $\ell_2$-$\epsilon$ heavy hitters.

\begin{definition} (Heavy Hitters Estimator) \label{HHproblem:def}
A $\ell_2$-$\epsilon$-heavy hitters estimator is
applied to a sketch of an input vector $\boldsymbol{v}\in\mathbb{R}^n$ and returns a set of entries $K\subset [n]$.  The output is {\em correct} if  $K$ includes all the $\ell_2$-$\epsilon$-heavy hitters keys and does not include keys with $v_i=0$.
\end{definition}

\begin{remark} \label{classic:rem}
This correctness definition is weaker than what the classic
estimator applied to $\countsketch$ provides~\cite{CharikarCFC:2002}:
A sketch with $\ell = O(\log (n/\delta))$ provides with probability $1-\delta$, approximate values for all $n$ keys that are within $\epsilon  \|\boldsymbol{v}_{\tail(1/\epsilon)} \|_2^2$, noting that the tail (vector $\boldsymbol{v}_{\tail(1/\epsilon)}$ that is our input $\boldsymbol{v}$ with $1/\epsilon$ heaviest
entries nullified) has smaller norm than
$\boldsymbol{v}$, yielding tighter estimates. In particular, this yields a correct solution to the heavy hitters problem with probability $1-\delta$.
Since our focus here is on designing an attack, our design is stronger against weaker requirements since less information on the randomness is being revealed. 
\end{remark}

\begin{definition} (Inner Product Estimator) \label{IPproblem:def}
An inner-product estimator is
applied to sketches of two input vectors $\boldsymbol{v},\boldsymbol{u}\in\mathbb{R}^n$ and returns $s\in \{-1,1\}$. The output is {\em correct}
if $s=-1$ when $\langle \boldsymbol{v},\boldsymbol{u}\rangle=0$ and is $s=1$ when
$\langle \boldsymbol{v},\boldsymbol{u}\rangle^2 \geq \epsilon \|\boldsymbol{v}\|_2^2 \|\boldsymbol{u}\|_2^2$.
\end{definition}  

\subsection{\texttt{\countsketch}}
The sketch~\cite{CharikarCFC:2002}  is specified by parameters $(n,\ell,b)$, 
where $n$ is the dimension of input vectors, and $d= \ell\cdot b$.
The internal randomness $\rho$ specifies a set of random hash functions 
$h_r:[n]\rightarrow [b]$ ($r\in [\ell]$) with the marginals that $\forall k\in [b]$, $i\in [n]$, $\Pr[h_j(i) = k]=1/b$, and 
$s_{j}:[n]\rightarrow \{-1,1\}$ ($ j\in [\ell]$) so that $\Pr[s_{j}(i)=1]=1/2$. 
These hash functions define $d = \ell\cdot b$ measurement vectors that are organized as $\ell$ sets of $b$ vectors each
$\boldsymbol{\mu}^{(j,k)}$ ($j\in [\ell]$, $k\in [b]$) where:
\[ \mu^{(j,k)}_i := \mathbbm{1}_{\{h_j(i) = k\}} s_{j}(i) .\]

For an input vector $\boldsymbol{v}\in\mathbb{R}^n$,
$\sketch_\rho(\boldsymbol{v}) := (\langle \boldsymbol{\mu}^{(j,k)},\boldsymbol{v}\rangle)_{j,k}$
is the set of the respective
measurement values.
Note that for each key $i\in[n]$ there are exactly $\ell$ measurement
vectors with a nonzero $i$th entry: $(\boldsymbol{\mu}^{(j,h_j(i))})_{j\in[\ell]}$ and these measurement vectors are independent (as the only dependency is between measurement in the same
set of $b$, and there is exactly one from each set).
The respective set of $\ell$ adjusted
measurements:
\begin{equation} \label{adjmeasure:eq}
(\langle \boldsymbol{\mu}^{(j,h_j(i))}, \boldsymbol{v} \rangle
\mu^{(j,h_j(i))}_i)_{j\in [\ell]}
\end{equation}  
are unbiased estimates of $v_i$:
$\E_\rho [\langle \boldsymbol{\mu}^{(j,h_j(i))}, \boldsymbol{v} \rangle
\mu^{(j,h_j(i))}_i ] = v_i$.

The median estimator~\cite{CharikarCFC:2002} uses the median adjusted measurement to estimate the value $v_i$ of each key $i$. The $O(1/\epsilon)$ keys with highest magnitude estimates are then reported as heavy hitters. 
When the same randomness $\rho$ is used for $r$ non-adaptive inputs
(inputs selected independently of $\rho$ and prior outputs), 
sketch parameter settings of $\ell=\log (r \cdot n/\delta)$ and $b =
O(\epsilon^{-1})$ guarantee that with probability $1-\delta$, all outputs are correct (in the sense of Remark~\ref{classic:rem}).

$\countsketch$ also supports estimation of inner products.
For two vectors $\boldsymbol{v}$, $\boldsymbol{u}$, we obtain an
unbiased estimate of their inner product from the respective
inner product of the $j\in [\ell]$th row of measurements:
\begin{equation} \label{innerproduct:eq}
\sum_{k\in[b]} \langle \boldsymbol{\mu}^{(j,k)}, \boldsymbol{v} \rangle \cdot
\langle \boldsymbol{\mu}^{(j,k)}, \boldsymbol{u} \rangle\ .
\end{equation}
The median of these $\ell$ estimates is within relative error $\sqrt{\epsilon}$ with probability $1-\exp(\Omega(-\ell))$.

We note that pairwise independent
hash functions $h_j$ and $s_{j}$ suffice for obtaining the guarantees of
Remark~\ref{classic:rem}
\cite{CharikarCFC:2002} whereas 4-wise independence is needed for approximate
inner products. The analysis of the attack we present here, however,
holds even under full randomness.

\subsection{Adversarial Input for \texttt{\countsketch}}

\begin{definition} [Adversarial input] \label{adversarialinput:def}
  We say that an attack $\mathcal{A}$ that is applied to a sketch with randomness $\rho$ and outputs
  $i\in [n]$ and $\boldsymbol{z}^{(A)}\in \{-1,0,1\}^n$ (with $z^{(A)}_i=0$) is   
  $(B,\beta)$-{\em adversarial} (for $B>1$) if, with probability at least $1-\beta$ over the randomness of $\rho,\mathcal{A}$, the adjusted measurements \eqref{adjmeasure:eq} satisfy:
  \begin{equation} \label{adverse:def}
      \Pr_{\rho\sim\mathcal{D},\mathcal{A}}\left[ \frac{1}{\ell} \sum_{j\in [\ell]} 
      \langle \boldsymbol{\mu}^{(j,h_j(i))}, \boldsymbol{z}^{(A)} \rangle
\mu^{(j,h_j(i))}_i \geq \sqrt{\frac{B}{b}}\| \boldsymbol{z}^{(A)}\|_2 \right] \geq 1-\beta .
  \end{equation}
\end{definition}

The adversarial input $\boldsymbol{z}^{(A)}$ is a noise vector (with no heavy hitters) but $\sketch_\rho(\boldsymbol{z}^{(A)})$ "looks like" (in terms of the average adjusted measurement of key $i$) a sketch of a vector with a heavy key $i$.  
It follows from the standard analysis of $\countsketch$ that the event
\begin{equation} \label{iheavy:eq}
      \frac{1}{\ell} \sum_{j\in [\ell]} 
      \langle \boldsymbol{\mu}^{(j,h_j(i))}, \boldsymbol{v} \rangle
\mu^{(j,h_j(i))}_i \geq \sqrt{\frac{B}{b}}\| \boldsymbol{v}\|_2\ .
  \end{equation}
(that correspondence to~\eqref{adverse:def}) is very likely for vectors $\boldsymbol{v}$ such that $h$ is a heavy hitter (Definition~\ref{HH:def}) and extremely unlikely when $h$ is not a heavy hitter, and in particular, when $v_h=0$.
Similarly for inner products,
considering  inner product 
$\boldsymbol{v}$
with the standard basis vector
$\boldsymbol{e}_i$, the 
sketch-based estimates \eqref{innerproduct:eq} of each $j\in \ell$ computed from $\sketch_\rho(\boldsymbol{e}_i)$ and
$\sketch_\rho(\boldsymbol{v})$
are equal to the respective 
adjusted measurement \eqref{adjmeasure:eq} of $i$ from 
$\sketch_\rho(\boldsymbol{v})$.
The event \eqref{iheavy:eq} is very likely when 
$\langle \boldsymbol{e}_i, \boldsymbol{v}\rangle \geq 
\epsilon \| \boldsymbol{v}\|_2^2$ and very unlikely when
$\langle \boldsymbol{e}_i, \boldsymbol{v}\rangle=0$, noting that
for our adversarial input $\boldsymbol{z}^{(A)}$ it holds that 
$\langle \boldsymbol{e}_i, \boldsymbol{z}^{(A)}\rangle =0$.

While this follows from the standard analysis, the simple structure of our noise vectors allows for a particularly simple argument for bounding the probability of \eqref{iheavy:eq} for $\boldsymbol{v}$ with the structure of $\boldsymbol{z}^{(A)}$ : The distribution of an adjusted measurement of a $\boldsymbol{v} \in \{-1,0,1\}^n$ and $v_i=0$ with support size
$m=|\supp{(\boldsymbol{v}}| = \|\boldsymbol{v}\|^2_2$ approaches $\mathcal{N}(0,\frac{m}{b})$ (for large $m/b$),  and thus
the average approaches
$\mathcal{N}(0,\frac{m}{\ell\cdot b})$. Therefore, the probability of \eqref{adverse:def} on a random sketch of $\boldsymbol{z}^{(A)}$ is $\leq \exp(-\ell B/2)$ (applying tail bounds on the probability of value exceeding $\sqrt{\ell B}$ standard deviations).

\section{Attack Description}
\label{sec:attack_structure}

We describe our attack against heavy hitters estimators.  The modifications needed for it to apply with inner product estimator are described in Section~\ref{IP:sec}.
 The attack is an interaction between the following components:
\begin{itemize}
    \item 
Internal randomness $\rho\sim\mathcal{D}$ that specifies linear measurement vectors $(\boldsymbol{\mu}^{(j,k)})_{j\in[\ell],k\in[b]}$. The sketch $\sketch_\rho(\boldsymbol{v})$ of a vector 
$\boldmath{v}\in \mathbb{R}^n$ is the set of measurements
$(\langle \boldsymbol{\mu}^{(j,k)},\boldsymbol{v}\rangle)_{j\in[\ell],k\in[b]}$.
 \item
A {\em query-response algorithm} that at each step $t$ chooses a heavy hitters estimator (see Definition~\ref{HHproblem:def}).
The choice may depend on the randomness $\rho$ and prior queries and responses $(\sketch_\rho(\boldsymbol{v}^{(t')}),K^{(t')})_{t'<t}$. 
The algorithm receives 
$\sketch_\rho(\boldsymbol{v}^{(t)})$, applies the estimator to the sketch,
and outputs $K^{(t)}$.
\item
An  {\em adversary} that issues a sequence of input queries  $(\boldsymbol{v}^{(t)})_t$ and collects the responses $(K^{(t)})_t$. 
The randomness $\rho$ and the sketches of the query vectors $(\sketch_\rho(\boldsymbol{v}^{(t)}))_t$ are not known to the adversary.
The goal is to construct an
  {\em adversarial input vector} $\boldsymbol{z}^{(A)}$ (see Definition~\ref{adversarialinput:def}).  
\end{itemize}

 Our adversary generates the
{\em query vectors} 
$(\boldsymbol{v}^{(t)})_{t\in [r]}$ non-adaptively
as described in Section~\ref{queryvec:sec}.  The attack interaction
and its properties are stated in Section~\ref{sec:universal_attack}.

\subsection{Query Vectors} \label{queryvec:sec}
Our attack query vectors $(\boldsymbol{v}^{(t)})_{t\in [r]}$ of
the form:
\begin{equation} \label{input:eq}
    \boldsymbol{v}^{(t)} := v^{(t)}_h \boldsymbol{e}_h +
\boldsymbol{z}^{(t)}\in \mathbb{R}^n, 
\end{equation}
where $h$ is a special {\em heavy} key, that is selected uniformly
$h\sim \mathcal{U}[n]$ and remains fixed, 
$\boldsymbol{e}_h$ is the standard basis vector (axis-aligned unit vector along $h$), and 
the vectors $\boldsymbol{z}^{(t)}$ are
{\em tails}.  The (randomized) construction of tails is described in
Algorithm~\ref{algo:tails}.
The tail vectors $(\boldsymbol{z}^{(t)})_{t\in[r]}$ have
support of size $|\supp(\boldsymbol{z}^{(t)})| = m$ that does not
include key $h$ 
($h\not\in  \supp(\boldsymbol{z}^{(t)})$) and so that the supports of different tails are disjoint:
\[
  t_1\not= t_2 \implies \supp(\boldsymbol{z}^{(t_1)}) \cap \supp(\boldsymbol{z}^{(t_2)}) =
  \emptyset\ .
\]
For query $t$ and key $i\in\supp(\boldsymbol{z}^{(t)})$, the values are  selected
i.i.d.\ Rademacher $z^{(t)}_i \sim \mathcal{U}[\{-1,1\}]$.
Note that
$\|\boldsymbol{v}^{(t)}\|_2^2 = (v_h^{(t)})^2+ \|\boldsymbol{z}^{(t)}\|_2^2 = (v_h^{(t)})^2+ m$.
Note that the tails, and (as we shall see)  the selection of
$v^{(t)}_h$,  and hence the input vectors are constructed
non-adaptively.  Only the
final adversarial input vector depends on the output of the estimator on prior queries.   The parameter $m$ is set to a value that is polynomial in the sketch size and large enough so that certain approximations hold (see Section~\ref{sketchdist:sec}).

\begin{algorithm2e}
{\small
    \caption{\texttt{AttackTails}}
    \label{algo:tails}
    \DontPrintSemicolon
    \KwIn{Input dimension $n$, support size $m$, number of tails $r$}
    $h \gets \mathcal{U}[n]$ \tcp*{Special heavy hitter key}
    $S\gets \{h\}$ \tcp*{Keys used in support}
    \For{$t\in [r]$}{
      $S' \gets $ random subset of size $m$ from $[n]\setminus S$\;
      $\boldsymbol{z}^{(t)} \gets \boldsymbol{0}$\;
      \ForEach{$i\in S'$}
      {  
$\boldsymbol{z}^{(t)}_i \sim \mathcal{U}[\{-1,1\}]$\;
        }
    $S \gets S\cup S'$
  }
    \Return{$h, (\boldsymbol{z}^{(t)})_{t\in [r]}$}
}
\end{algorithm2e}

\begin{remark}
The only piece of information needed from the output of the estimator is whether the particular key $h$ is reported as a candidate heavy hitter of $\boldsymbol{v}^{(t)}$, that is, whether $h\in K^{(t)}$. 
Note that disclosing additional information
can only make the estimator {\em more} vulnerable to attacks.
\end{remark}

\subsection{Universal Attack}
\label{sec:universal_attack}


\begin{algorithm2e}
{\small
    \caption{\texttt{Universal Attack on $\countsketch$} Heavy Hitters Estimators}
    \label{algo:attack}
    \DontPrintSemicolon
    Use $a = \Theta(\sqrt{\frac{\ln(1/\delta_2)}{\ell}})$
and $c = \Theta(1)$\tcp*{With universal constants as in Lemma~\ref{correctest:lemma}}
    \KwIn{Initialized $\countsketch_\rho$ with parameters $(n,\ell,b)$, Query-response algorithm, number of queries $r$, tail support size $m$}
    $(h,(\boldsymbol{z}^{(t)})_{t\in [r]}) \gets$
    \texttt{AttackTails}$(n,m,r)$\tcp*{Algorithm~\ref{algo:tails}}
 \For(\tcp*[h]{Compute Query Vectors}){$t\in [r]$}{
    $v^{(t)}_h\sim \mathcal{U}[a\cdot \sigma, (c+2a)\cdot \sigma]$\tcp*{$\sigma := \sqrt{m/b}$}
    $\boldsymbol{v}^{(t)} \gets  v^{(t)}_h \boldsymbol{e}_h +
    \boldsymbol{z}^{(t)}$\tcp*{Query vectors}
    }
    \For(\tcp*[h]{Apply Query Response}){$t\in [r]$}
    {
      Choose a correct HH estimator $M^{(t')} $\tcp*{With correct reporting function (Definition~\ref{correctmean:def}), may depend on $(v^{(t')}_h,K^{(t')},M^{(t')})_{t'<t}$ and $\rho$ }
       $K^{(t)} \gets M^{(t)}(\countsketch_\rho(\boldsymbol{v}^{(t)}))$\tcp*{Apply estimator to sketch}
    \leIf{$h\in K^{(t)}$}{$s^{(t)}\gets 1$}{$s^{(t)}\gets -1$}
      }
     \Return{$\boldsymbol{z}^{(A)} \gets \sum_{t\in[r]} s^{(t)} \boldsymbol{z}^{(t)} $}
    \tcp*{Adversarial input}
    }
  \end{algorithm2e}
  
  Our attack interaction is described in Algorithm~\ref{algo:attack}.  We generate $r$ attack tails 
  using Algorithm~\ref{algo:tails}.  We then construct $r$ queries of the form \eqref{input:eq} with 
  i.i.d.\ $v^{(t)}_h\sim \mathcal{U}[a\cdot \sigma, (c+2a)\cdot \sigma]$.  At each step $t\in[r]$, we feed the sketch of $\boldsymbol{v}^{(t)}$ to the HH estimator selected by the query response algorithm and collect the output $K^{(t)}$ of the  estimator. We then set $s^{(t)} \gets 1$ if $h\in K^{(t)}$ and $s^{(t)} \gets -1$ if $h\not\in K^{(t)}$. 
 The final step computes the 
  {\em adversarial input}:
 \begin{equation}\label{adversetail:eq}
\boldsymbol{z}^{(A)} := \sum_{t\in[r]} s^{(t)} \boldsymbol{z}^{(t)} \ .
\end{equation}

The statements below only use measurement vectors with nonzero value for key $h$.  To simplify the notation, we
use $\boldsymbol{\mu}^{(j)} := \boldsymbol{\mu}^{(j,h_j(h))}$ for $j\in[\ell]$.
For randomness $\rho$, we use the notation $\boldsymbol{\mu}^{(j)}(\rho)$ for the respective measurement
vectors and $\mathcal{A}(\rho)$ for the output distribution of Algorithm~\ref{algo:attack} applied with randomness $\rho$. 

The adversarial input has $z^{(A)}_h = 0$ and norm $\|\boldsymbol{z}^{(A)} \|^2_2 = r\cdot m$ (it has support of size $r\cdot m$ with values in the support i.i.d\ Rademacher $U[\{-1,1\}]$).  
When an adversarial input $\boldsymbol{z}^{(A)} \sim \mathcal{A}(\rho_0)$ is sketched using 
a {\em random} $\rho\sim\mathcal{D}$ it holds that:
\begin{align*}
\forall j\in[\ell],\ \E_{\rho_0\sim \mathcal{D},\boldsymbol{z}^{(A)}\sim \mathcal{A}(\rho_0),\rho\sim\mathcal{D}}\left[\langle \boldsymbol{z}^{(A)}, \boldsymbol{\mu}^{(j)}(\rho)\rangle\right] &=0\ \\
\forall j\in[\ell],\ \Var_{\rho_0\sim \mathcal{D},\boldsymbol{z}^{(A)}\sim \mathcal{A}(\rho_0),\rho\sim\mathcal{D}}\left[\langle \boldsymbol{z}^{(A)}, \boldsymbol{\mu}^{(j)}(\rho) \rangle\right] &\approx \frac{r\cdot m}{b} = r\sigma^2\ .
\end{align*}
and since the $\ell$ measurements are independent we get:
\begin{align*}
     \E_{\rho_0\sim \mathcal{D},\boldsymbol{z}^{(A)}\sim \mathcal{A}(\rho_0),\rho\sim\mathcal{D}} \left[\frac{1}{\ell} \sum_{j\in [\ell]} \langle \boldsymbol{z}^{(A)}, \boldsymbol{\mu}^{(j)}(\rho)\rangle \cdot \mu^{(j)}_h(\rho)\right] &= 0\\
      \Var_{\rho_0\sim \mathcal{D},\boldsymbol{z}^{(A)}\sim \mathcal{A}(\rho_0),\rho\sim\mathcal{D}} \left[\frac{1}{\ell} \sum_{j\in [\ell]} \langle \boldsymbol{z}^{(A)}, \boldsymbol{\mu}^{(j)}(\rho)\rangle \cdot \mu^{(j)}_h(\rho)\right] &\approx \frac{r}{\ell} \cdot \sigma^2
  \end{align*}
 
 The adversarial input $\boldsymbol{z}^{(A)} \sim \mathcal{A}(\rho_0)$ behaves differently with respect to the particular randomness $\rho_0$ it was constructed for.  We will establish the following:

 \begin{lemma} [Properties of the adversarial input] \label{tailprop:lemma}
 \begin{align*}
     \E_{\rho_0\sim \mathcal{D},\boldsymbol{z}^{(A)}\sim \mathcal{A}(\rho_0)}\left[\frac{1}{\ell}\sum_{j\in [\ell]} \langle \boldsymbol{z}^{(A)}, \boldsymbol{\mu}^{(j)}(\rho_0)\rangle \cdot \mu^{(j)}_h(\rho_0)\right] &\approx \frac{r}{\ell} \cdot \frac{2\sigma}{c+a}\\
     \Var_{\rho_0\sim \mathcal{D},\boldsymbol{z}^{(A)}\sim \mathcal{A}(\rho_0)}\left[\frac{1}{\ell}\sum_{j\in [\ell]} \langle \boldsymbol{z}^{(A)}, \boldsymbol{\mu}^{(j)}(\rho_0)\rangle \cdot \mu^{(j)}_h(\rho_0)\right] &\approx \frac{r}{\ell} \cdot \sigma^2
  \end{align*}
  \end{lemma}
  The proof is provided in Section~\ref{simplifiedminest:sec} (following preparation in Sections~\ref{sketchdist:sec}-\ref{sec:estimators} and details in Section~\ref{minestproofs:sec}). The high level idea, as hinted in the introduction,  is that we establish that $h\in K^{(t)}$ and thus   $s^{(t)}=1$ is correlated with "positive bias", that is, with the event
  $\frac{1}{\ell} \sum_{j\in [\ell]} \langle s^{(t)}\boldsymbol{z}^{(t)}, \boldsymbol{\mu}^{(j)}(\rho_0)\rangle \cdot \mu^{(j)}_h(\rho_0) > 0$. In the sum
  $\sum_{t\in[r]} s^{(t)} \boldsymbol{z}^{(t)}$ the bias (which is ``forced''  error on the estimates) increases linearly with $r$ while the $\ell_2$ norm, which corresponds to the standard deviation of the error, increases
proportionally to $\sqrt{r}$.

With Lemma~\ref{tailprop:lemma} in hand, we are ready to show that $\boldsymbol{z}^{(A)}$ is an adversarial input (see Definition~\ref{adversarialinput:def}). 
\begin{theorem}[Adversarial input]
If for $B>1$ we use attack of size $r = B\cdot \ell^2 $ then 
\begin{align*}
     \E_{\rho_0\sim \mathcal{D},\boldsymbol{z}^{(A)}\sim \mathcal{A}(\rho_0)}\left[\frac{1}{\ell} \sum_{j\in [\ell]} \langle \boldsymbol{z}^{(A)}, \boldsymbol{\mu}^{(j)(\rho_0)}\rangle \cdot \mu^{(j)}_h(\rho_0)\right] &\approx \frac{2}{c+a} \sqrt{\frac{B}{b}} \| \boldsymbol{z}^{(A)} \|_2\\
      \Var_{\rho_0\sim \mathcal{D},\boldsymbol{z}^{(A)}\sim \mathcal{A}(\rho_0)}\left[\frac{1}{\ell} \sum_{j\in [\ell]} \langle \boldsymbol{z}^{(A)}, \boldsymbol{\mu}^{(j)}(\rho_0)\rangle \cdot \mu^{(j)}_h(\rho_0)\right] &\approx \frac{1}{\ell\cdot b} \|\boldsymbol{z}^{(A)} \|^2_2\ .
\end{align*}     
\end{theorem}
\begin{proof}
Using Lemma~\ref{tailprop:lemma},
the expectation of the average with attack size $r=B\ell^2$ is
\begin{align*}
    \frac{r}{\ell} \cdot \frac{2\sigma}{c+a} &= \frac{r}{\ell} \sqrt{\frac{m}{b}} \cdot \frac{2}{c+a} \\
    &= \frac{2}{c+a} \frac{\sqrt{r}}{\ell} \frac{1}{\sqrt{b}}  \sqrt{r\cdot m} \\
    &= \frac{2}{c+a} \sqrt{\frac{B}{b}} \|\boldsymbol{z}^{(A)} \|_2 & \text{since $\|\boldsymbol{z}^{(A)} \|_2 = \sqrt{r\cdot m}$ }
\end{align*}
The variance of the average is
\begin{align*}
    \frac{r}{\ell} \sigma^2 =  \frac{r\cdot m}{\ell\cdot b}= \frac{1}{\ell\cdot b} \|\boldsymbol{z}^{(A)} \|^2_2
\end{align*}
\end{proof}

\subsection{Attack with an inner-product estimator} \label{IP:sec}

We describe the modifications to Algorithm~\ref{algo:attack} needed for the attack 
to apply with an inner-product estimator. 
We compute the same query vectors 
$\boldsymbol{v}^{(t)}_{t\in [r]}$.
At each step $t$, the query response algorithm chooses a correct inner-product estimator $M^{(t)}$
(see Definition~\ref{IPproblem:def}).
The query is issued for the inner 
product of 
$\boldsymbol{v}^{(t)}$
with the standard basis vector
$\boldsymbol{e}_h$.  Note that the value of the inner product is exactly $v^{(t)}_h$ and the requirement of correct reporting of the inner product (Definition~\ref{IPproblem:def}) on these query vectors matches the requirement of a correct heavy hitters reporting of the key $h$ (Definition~\ref{HHproblem:def}).

The input to the estimator are the sketches
$\sketch_\rho(\boldsymbol{e}_h)$ and 
$\sketch_\rho(\boldsymbol{v}^{(t)})$.
Note that the information available to the estimator from the provided sketches on $v^{(t)}_h$ is the same as with heavy-hitter queries:  $\sketch_\rho(\boldsymbol{e}_h)$ is simply the vector with entries $\mu_h^{(j)}$, which does not add information as $\rho$ and $h$ are assumed to be known to the estimator.
The same analysis therefore applies. 

\section{Sketch distribution} \label{sketchdist:sec}

In this section we show that with our particular query inputs \eqref{input:eq}, for large enough $m$, the sketch content that is relevant to determining whether $h$ is a candidate heavy hitter is approximately
$\boldsymbol{u}^{(t)} \sim \mathcal{N}_\ell(v^{(t)}_h,\sigma^2)$ where $\sigma = \sqrt{\frac{m}{b}}$. The random variables
$\boldsymbol{u}^{*(t)} = \boldsymbol{u}^{(t)}- v^{(t)}_h \boldsymbol{1}_\ell$ for $t\in[r]$ are approximately  i.i.d.\ from $\mathcal{N}_\ell(0,\sigma^2)$.

$\countsketch_\rho(\boldsymbol{v}^{(t)})$ includes $\ell\cdot b$
measurements but for inputs known to be of a form \eqref{input:eq} we
can restrict the estimator to only consider the $\ell$ adjusted
measurements \eqref{adjmeasure:eq} of key $h$.  We argue below that this 
restriction does not limit the power of the estimator.

To simplify our notation going forward, we re-index the set of
relevant measurement vectors and use 
$(\boldsymbol{\mu}^{(j)})_{j\in[\ell]}$ for
$(\boldsymbol{\mu}^{(j,h_j(h))})_{j\in[\ell]}$.
We denote by
$\boldsymbol{u}^{(t)}\in\mathbb{R}^\ell$
the random variable vector of the $\ell$ adjusted measurements (see Definition~\ref{adjmeasure:eq})
the sketch provides for $v^{(t)}_h$ and a random tail and by
$\boldsymbol{u}^{*(t)}\in\mathbb{R}^\ell$ the random variable of the respective contributions of the tail component $\boldsymbol{z}^{(t)}$ to these values:
 \begin{eqnarray} 
u^{*(t)}_j :=&  \langle \boldsymbol{z}^{(t)}, \boldsymbol{\mu}^{(j)}(\rho_0) \rangle \cdot  \mu^{(j)}_h(\rho_0) \label{ustardef:eq}\\ 
u^{(t)}_j :=& \langle \boldsymbol{v}^{(t)}, \boldsymbol{\mu}^{(j)}(\rho_0) \rangle \cdot  \mu^{(j)}_h(\rho_0) = v^{(t)}_h + u^{*(t)}_j \ . \label{onebucket:eq}
\end{eqnarray}

\begin{lemma}
The distribution of $(\boldsymbol{u}^{*(t)})_{t\in [r]}$ 
(over $\rho\in \mathcal{D}$ and random tails selection) is such that
$u^{*(t)}_j$ are i.i.d.\ 
for $t\in [r]$, $j\in [\ell]$) with distribution 
\begin{align}
    M^{(t)}_j &\sim \mathcal{B}(m,1/b) \nonumber\\
u^{*(t)}_j &\sim  2\cdot \mathcal{B}(M^{(t)}_j,\frac{1}{2})-M^{(t)}_j\ .    \label{ustardist:eq}
\end{align}
\end{lemma}
\begin{proof}
Note that the distribution of 
$(\boldsymbol{u}^{*(t)})_{t\in [r]}$ 
does not depend on the choices of
$(v^{(t)}_h)_{t\in [r]}$.

Consider the distribution of 
$u^{*(t)}_j$ (for $t\in [r]$, $j\in [\ell]$). The respective 
measurement vector is $\boldsymbol{\mu}^{(j)}$. 
We have $|\supp(\boldsymbol{z})|=m$ and for each $i\in \supp(\boldsymbol{z})$,  $\mu^{(j)}_i\not=0$ with probability $1/b$. Therefore, the 
number of keys in the support that contribute to the
measurement 
\[
M^{(t)}_j := \left| \{i \mid z^{(t)}_i \cdot \mu^{(j)}_i \not= 0\}\right|
\]
is a Binomial random variable $M^{(t)}_j \sim \mathcal{B}(m,1/b)$, hence, has expectation $\frac{m}{b}$ and variance $\frac{m}{b}$. The contributions of the $M^{(t)}_j$ keys are i.i.d.\ Rademacher $U\{-1,1\}$ (multiplying by $\mu^{(j)}_h$ which does not change that). 
Therefore, the contribution to the measurement, conditioned on $M^{(t)}_j$, has distribution 
\begin{equation*}
u^{*(t)}_j \sim  2\cdot \mathcal{B}(M^{(t)}_j,\frac{1}{2})-M^{(t)}_j\ .
\end{equation*}
Note that the random variables
$u^{*(t)}_j$ are symmetric and
$\E[u^{*(t)}_j]=0$.

Recall that the values $\mu^{(j)}_i$ for all keys $i\not=h$ and $j\in [\ell]$ are i.i.d.\ ($0$ with probability $(b-1)/b)$ and Rademacher otherwise).
Recall that the random tails are independent of the sketch randomness and the supports $\supp(\boldsymbol{z}^{(t)})$ for $t\in [r]$ are disjoint.  Even when conditioned on fixing the supports for all $t\in [r]$, all values of keys within the support, across all $t\in [r]$, are i.i.d.\ Rademacher.

We now argue that 
$\boldsymbol{u}^{*(t)}$ are i.i.d.\ for different $t\in [r]$.  This follows from the disjoint supports of
$\supp(\boldsymbol{z}^{(t)})$ for different $t\in [r]$ and therefore
\eqref{ustardef:eq} for different $t$ only depend on disjoint parts of each set of measurement vectors
$(\boldsymbol{\boldsymbol{\mu}^{(j)}})_{j\in [\ell]}$, with only values for keys $i\in \supp(\boldsymbol{z}^{(t)})$.

We next observe that for any $t\in [r]$,  the values
$u^{*(t)}_j$ are i.i.d.\ for $j\in [\ell]$.
In this case the applicable support $\supp(\boldsymbol{z}^{(t)})$
is fixed.  But the applicable parts of different measurement vectors $j\in[\ell]$ are independent.
Even when there is overlap in the sets of keys other than $h$ in $\supp(\boldsymbol{z}^{(t)})$ that map to different measurements, the measurement vectors randomness is independent, so the contributions of the same key to different measurements it maps to are independent $U\{-1,1\}$.
\end{proof}

We are now ready justify considering only the measurement with $\mu_h\not=0$ in each set $j\in [\ell]$:
The additional measurements have $\boldsymbol{\mu}_h= 0$ and thus do not
depend on $v_h$.  They do provide information on 
the tail support size parameter $m$ but we can assume $m$ is known to
the estimator. These measurements also provide information on $M^{(t)}_j$, the number of keys in the support that hash to our selected measurement from the each set of $b$ measurements. 
This because their sum has distribution
$\mathcal{B}(\mathcal{B}(m-M^{(t)}_j,1/b),1/2)$.  But in our analysis we assume that
$m$ is chosen to be large enough so that $M^{(t)}_j$ is very close to its expectation $m/b$ and we may also assume that the exact values $M^{(t)}_j$ are available to the estimator (Essentially, the variations in the value only impacts the "noise" magnitude but not its form).

\begin{lemma}
The distribution of the random variable $\sum_{j\in [\ell]}u^{*(t)}_j$ conditioned on $\rho$ and $\boldsymbol{M}^{(t)}$ is
\[
2 \mathcal{B}\left( \sum_{j\in[\ell]} M^{(t)}_j,1/2\right) - \sum_{j\in[\ell]} M^{(t)}_j\ .
\]
\end{lemma}
\begin{proof}
Taking the sum over $j\in[\ell]$ of the distributions of $u^{*(t)}_j$
\begin{align*}
    &\sum_{j\in[\ell]} \left( 2 \mathcal{B}(M^{(t)}_j,1/2)- M^{(t)}_j\right) \\
     =& 2 \sum_{j\in[\ell]} \mathcal{B}(M^{(t)}_j,1/2)- \sum_{j\in[\ell]} M^{(t)}_j \\
     =& 2\mathcal{B}( \sum_{j\in[\ell]} M^{(t)}_j,1/2) - \sum_{j\in[\ell]} M^{(t)}_j\ ,
\end{align*}
where the last equality follows from properties of a sum of independent Binomial random variables. 
\end{proof}

\begin{remark} \label{mchoice:rem}
We choose $m$ to be large enough so that 
the following approximations hold:
\begin{itemize}
    \item [(i)]
The independent random variables 
$M^{(t)}_j \sim \mathcal{B}(m, 1/b)$ are close to their mean $m/b$ for all $t\in[r]$ and $j\in [\ell]$: More precisely, for some $\alpha$ and $\delta'$, with probability $(1-\delta')$, for all $t,j$, $(M^{(t)}_j \in 1\pm \alpha) m/b$. We bound this using Chernoff bounds. We solve
$2\exp(- \frac{m}{b} \alpha^2/3 \leq \delta'/(r \ell )$ to obtain
$m \geq 3 b \alpha^{-2} \ln(2r\ell/\delta')$.
\item [(ii)]
 With probability $(1-\delta')$, $r \ell$ random variables
$\mathcal{B}(n,1/2)$ where $n$ is in $(1\pm \alpha) m/b$, are all within
$O(n^{2/3} \alpha)$ from their expected value $n/2$.
When this holds, the normal approximation to the Binomial distribution has relative error within $1\pm \alpha$.
We again apply Chernoff bounds solving for
$2\exp(- (n/2) (n^{-1/3} \alpha)^2/3 \leq \delta'/(r \ell ) $
solving we get
$n^{1/3} \geq 6\alpha^{-2} \ln(2r\ell/\delta')$ and therefore
$m = O(b \alpha^{-6} \ln(2r\ell/\delta')^3)$.
\end{itemize}
\end{remark}

When $m/b$ is large in the sense of Remark~\ref{mchoice:rem}, the distribution of  $u^{*(t)}_j$ is approximated well by $\mathcal{N}(0, m/b)$  and 
$\boldsymbol{u}^{*(t)}$ is approximated by $\mathcal{N}_\ell(0, m/b)$
and therefore
\begin{eqnarray}
\boldsymbol{u}^{(t)} &\sim& \mathcal{N}_\ell(v^{(t)}_h,\frac{m}{b})  \label{udist:eq}\\
\overline{u} &\sim& \mathcal{N}(v^{(t)}_h, \frac{m}{b} \cdot\frac{1}{\ell})\ .
\end{eqnarray}
The approximation is up to a 
relative error of
$\alpha=\alpha'/\ell$ and discretization.  We use $\alpha = \alpha'/\ell$ so that a relative error approximation of $1\pm \alpha'$ holds also with respect to a product distribution of $\ell$ such random variables.
For the purpose of cleaner presentation, we will use this approximation going forward. We use
$\sigma := \sqrt{\frac{m}{b}}$.
The analysis carries over by carefully dragging the approximation error term that vanishes for large $m$.

\section{Estimators}
\label{sec:estimators}

In this section we 
provide a framework of estimators and 
establish properties common to any correct $\ell_2$-$\epsilon$ heavy hitters estimator that is applied to our query vectors.

In its most general form, a query response algorithm fixes before each query $t$ an estimator $M^{(t)}$.
The estimator is applied to the content of the sketch, which on our inputs are i.i.d.\ vectors
$(\boldsymbol{u}^{(t)}\sim \mathcal{N}_\ell(v^{(t)}_h,\sigma^2))_{t\in[r]}$. The estimator is
specified by a {\em reporting function}
$p^{(t)}: \mathbb{R}^\ell \to [0,1]$ so that
$p^{(t)}(\boldsymbol{u}^{(t)}) := \Pr[h\in M^{(t)}(\boldsymbol{u}^{(t)})]$
specifies the probability that the returned $K^{(t)}$ includes key $h$ when the sketch content is $\boldsymbol{u}^{(t)}$.  We allow the query response algorithm to modify the estimator
arbitrarily between queries and in a way that depends on sketches of prior inputs, prior outputs, and on a maintained state from past queries $(\boldsymbol{u}^{(t')},K^{(t')})_{t'<t}$.
The only constraint that we impose is that (at each step $t$) the output is correct with high probability: 
$\ell_2$-$\epsilon$-heavy hitters are reported and $0$ value keys are not reported (see Definition~\ref{HHproblem:def}).
We show that a correct estimator on our query inputs must satisfy the following:
\begin{lemma} [Correct HH estimator basic property] \label{correctest:lemma}
For $\delta_1,\delta_2 \ll 1$, 
there are 
$a = \Theta(\sqrt{\frac{\ln(1/\delta_2)}{\ell}})$
and $c = \Theta(1)$ so that the following holds.
If the  
estimator satisfies (i) if $h$ is a heavy hitter then $\Pr[h\in K] \geq 1-\delta_1$ and (ii) if
$v_h=0$ then $\Pr[h\not\in K] \geq 1-\delta_2$.
Then 
\begin{itemize}
    \item $|v_h| \geq c\cdot\sigma$ $\implies$ $\Pr[h\in K] \geq 1-\delta_1$
    \item $|v_h| \leq a\cdot\sigma$ $\implies$ $\Pr[h\not\in K] \geq 1- \frac{1}{\delta_2^{\Omega(1)}}$
    \item Otherwise, unrestricted 
\end{itemize}
\end{lemma}
\begin{proof}
The tail vectors have $\ell_2^2$ mass $\|\boldsymbol{z} \|^2_2=m$ and
 $\|\boldsymbol{v}^{(t)} \|^2_2=m + (v^{(t)}_h)^2$.  Key $h$ is $\ell_2$ $\epsilon$-heavy hitter (see Definition~\ref{HH:def}) when
\begin{align*}
        v_h^2 \geq \epsilon\cdot \|\boldsymbol{v}\|_2^2 = \epsilon \cdot (v_h^2 + \|\boldsymbol{z}\|_2^2) \implies \\
        v_h^2 \geq \frac{\epsilon}{1-\epsilon} \|\boldsymbol{z}\|_2^2 =  \frac{\epsilon}{1-\epsilon} m = \frac{b \epsilon}{1-\epsilon}\cdot \sigma \ .
  \end{align*} 
 Recall that the sketch uses $b= O(1/\varepsilon)$, so when $c\approx \frac{b \epsilon}{1-\epsilon} = \Theta(1)$ we get that $|v_h| \geq c\cdot\sigma$ implies that $h$ is $\ell_2$ $\epsilon$-heavy hitter.  Therefore an estimator that reports a heavy key with probability at least $1-\delta_1$ must have $\Pr[h\in K]\geq 1-\delta_1$.
 
 We now establish the second part.  Here we show that if the probability of reporting when $v_h=0$ is at most $\delta_2$ then the probability of reporting 
when $|v_h|\leq a\cdot \sigma$ is at most $\delta_2^{1/4}$.
The probability density of a sketch $\boldsymbol{u}$ when $v_h=v$ is:
    \[
f_v(\boldsymbol{u}) = 
\prod_{i\in[\ell]} \frac{1}{\sigma\sqrt{2\pi}} e^{-\frac{1}{2} \left(\frac{u_i-v}{\sigma}\right)^2}\ ,
\]
We seek to bound the maximum the maximum reporting probability when $v_h \leq a$ subject to reporting probability of $0$ that is at most $\delta_2$.  Consider the ratio
\[
\frac{f_v(\boldsymbol{u})}{f_{0}(\boldsymbol{u})} = e^{-\frac{1}{2 \sigma^2}\left( \ell\cdot v^2 - 2\cdot v\cdot \sum_i u_i \right)} = e^{-\frac{1}{2 \sigma^2}\left( \ell\cdot (v^2 - 2 v \overline{u}) \right)}
\]
Note that the ratio depends only on $\overline{u}$ and $v$ and is larger when 
$\overline{u}$ is closer to $v$ than to $0$. Therefore, for the goal of bounding the maximum reporting probability with $v$ subject on a bound on the probability for $0$, we place
the reporting probability when $0$ on the largest values of $\overline{u}$. The distribution of $\overline{u}$ with $0$ mean is $\mathcal{N}(0,\sigma^2/\ell)$.
For $\overline{u} \geq \frac{\sigma}{\sqrt{\ell}} \sqrt{\ln(1/\delta_2)}$ provides reporting probability $\Theta(1/\delta_2)$ for $0$.  For $v = \frac{1}{2} \frac{\sigma}{\sqrt{\ell}} \sqrt{\ln(1/\delta_2)}$ the reporting probability is $\frac{1}{\delta_2^{\Omega(1)}}$.

\end{proof}

\section{Simplified Mean Estimation}\label{simplifiedminest:sec}

In this section we formulate and state properties of an interaction, described as
Algorithm~\ref{algo:meanestattack}, 
with a mean estimator of i.i.d\ random variables.  
We then show a correspondence between
Algorithm~\ref{algo:attack} 
and
Algorithm~\ref{algo:meanestattack}
establish properties of Algorithm~\ref{algo:attack} through this correspondence. In Section~\ref{AMS:sec} we apply this formulation to analyse an attack on norm estimators applied to the AMS sketch.

We use the following definition:
\begin{definition} [correct reporting function] \label{correctmean:def}
A reporting function
$p: \mathbb{R}^\ell \to [0,1]$ 
is {\em correct} with respect to
parameters: $(\delta,a,c,\ell,\sigma)$ if
\begin{align*}
    \forall v\geq c\cdot\sigma,\ 
\E_{\boldsymbol{u} \sim \mathcal{N}_\ell(v,\sigma^2)} [ p^{(t)}(\boldsymbol{u})] \geq 1-\delta\\
    \forall v\leq a\cdot\sigma,\ 
\E_{\boldsymbol{u} \sim \mathcal{N}_\ell(v,\sigma^2)} [ p^{(t)}(\boldsymbol{u})] \leq \delta\ .
\end{align*}
\end{definition}

A correct reporting function can be viewed as a simple mean estimator applied to $\ell$ i.i.d.\ samples from $\mathcal{N}(v,\sigma^2)$:
With  probability $1-\delta$, the output is $1$ when
$|v|> c\cdot\sigma$ and $-1$ when $|v| < a\cdot\sigma$.

We will establish the following (proof provided in Section~\ref{minestproofs:sec}):
\begin{lemma} [\texttt{MeanEstAttack} Properties]
\label{biasminest:lemma}
Consider Algorithm~\ref{algo:meanestattack}  where
$b$ is such that
\[
\sqrt{\frac{\ell}{2\pi}} e^{-\ell b^2/2} \ll \frac{1}{c-a+2b} .
\]
Then the output $\boldsymbol{u}^{(A)} \in \mathbb{R}^\ell$  satisfies
 \begin{align}
     \E_{\mathcal{A}}\left[\overline{u^{(A)}}\right]
       &\approx \frac{r}{\ell} \cdot \frac{2 \sigma}{c-a+2b} \label{biasB:eq} \\
     \frac{r}{\ell}\sigma^2 \left(1 - \frac{2b}{c-a+2b}  \right)^2 &\lessapprox \Var_{\mathcal{A}}\left[\overline{u^{(A)}}\right] \lessapprox \frac{r}{\ell}\sigma^2 \label{biasvarB:eq}
  \end{align}
  \end{lemma}

\begin{algorithm2e}
{\small
    \caption{\texttt{MeanEstAttack}}
    \label{algo:meanestattack}
    \DontPrintSemicolon
    \KwIn{Parameters $(a,c,\sigma,\delta)$, number of queries $r$, $b\in (0,a]$, A query response algorithm $\mathcal{A}$ that chooses $(a,c,\sigma,\delta)$-correct reporting functions}
    \For(\tcp*[h]{Generate queries and reporting functions}){$t\in [r]$}
    {
      $\mathcal{A}$ chooses a $(a,c,\sigma,\delta)$-correct reporting function $p^{(t)}:\mathbb{R}^\ell\to [0,1]$ \tcp*{Definition~\ref{correctmean:def}. Choice may depend on $(\boldsymbol{u}^{(t')},s^{(t')})_{t'<t}$}
      $v^{(t)}\sim \mathcal{U}[a\cdot \sigma, (c+2b)\cdot \sigma]$\;
      $\boldsymbol{u}^{*(t)} \sim \mathcal{N}_\ell(0,\sigma^2)$\;
      \lFor(\tcp*[h]{Compute $\boldsymbol{u}^{(t)}\in\mathbb{R}^\ell$, note that $\boldsymbol{u}^{(t)} \sim \mathcal{N}_\ell(v^{(t)},\sigma^2)$}){$j\in [\ell]$}{$u^{(t)}_j \gets v^{(t)} + u^{*(t)}_j$}
      $s^{(t)} \gets 1$ w.p.\ $p^{(t)}(\boldsymbol{u}^{(t)})$ and $s^{(t)} \gets -1$ otherwise
      }
     \Return{$\boldsymbol{u}^{(A)} \gets \sum_{t\in[r]} s^{(t)} \boldsymbol{u}^{*(t)} $}
        }
  \end{algorithm2e}

\subsection{Correspondence between the attack on \texttt{CountSketch} and \texttt{MeanEstAttack}}
We state the attack on $\countsketch$  in
Algorithm~\ref{algo:attack} in terms of 
Algorithm~\ref{algo:meanestattack}.

We first relate the
 correctness of an estimator to  correctness of its reporting function (proof is immediate):
 \begin{lemma} \label{correctp:coro}
If for parameters
$\delta,a,c,\ell,\sigma$ an estimator satisfies (i)
$|v_h| \geq c\cdot\sigma$ $\implies$ $\Pr[h\in K] \geq 1-\delta$ and
(ii) $|v_h| \leq a\cdot\sigma$ $\implies$ $\Pr[h\not\in K] \geq 1- \delta$, then its reporting function $p$ is correct 
with respect to parameters 
$(\delta,a,c,\ell,\sigma)$.
\end{lemma}
It follows that an estimator that satisfies the conditions of Lemma~\ref{correctest:lemma} has a reporting function that is correct as in Definition~\ref{correctmean:def} for parameters $(\max\{\delta_1,\delta_2^{\Omega(1)}\},a,c,\ell,\sigma)$.

In the correspondence between
 Algorithm~\ref{algo:attack} and
Algorithm~\ref{algo:meanestattack},
the reporting function  $M^{(t)}$ corresponds to 
$p^{(t)}$,  the values $v^{(t)}_h$  to $v^{(t)}$, 
the adjusted measurements of key $h$ from the sketch of  $\boldsymbol{v}^{(t)}$
correspond to $\boldsymbol{u}^{(t)}$, the adjusted measurements of
key $h$ from the sketch of  $\boldsymbol{z}^{(t)}$ corresponds to $\boldsymbol{u}^{*(t)}$, and  $s^{(t)}$ plays the same role. 
The adjusted measurements of the output of Algorithm~\ref{algo:attack}, adversarial vector
$\boldsymbol{z}^{(A)}$, correspond to 
the output $\boldsymbol{u}^{(A)}$ of Algorithm~\ref{algo:meanestattack}.
The parameter values $a,c$ are the same and the parameter $b$ in
Algorithm~\ref{algo:meanestattack} 
(unrelated to the sketch parameter) is set to $b\gets a$.

Using this correspondence, we can now provide a proof of
Lemma~\ref{tailprop:lemma}
(Properties of the adversarial tail).
The Lemma follows as a Corollary of Lemma~\ref{biasminest:lemma} using the established relation between Algorithm~\ref{algo:attack} and Algorithm~\ref{algo:meanestattack} and $b=a =\Theta(\sqrt{\ln(1/\delta_2)/\ell}$.
 When $\delta_2$ as used to determine $a$ in Lemma~\ref{correctest:lemma} is small enough,  we have
 $\sigma\cdot \sqrt{\frac{2}{\pi \ell}} e^{- \ell a^2/2} \approx \sigma \cdot \sqrt{\frac{2\cdot \delta_2}{\pi \ell}}$ and the condition of Lemma~\ref{biasminest:lemma} holds.

 \begin{proof} [Proof of Lemma~\ref{tailprop:lemma}]
  The contributions of the adversarial tail to the adjusted estimates in the sketch are:
\begin{align*}
  \langle \boldsymbol{z}^{(A)}, \boldsymbol{\mu}^{(j)}\rangle \cdot \mu^{(j)}_h =
  \sum_{t\in [r]} s^{(t)} \langle \boldsymbol{z}^{(t)}, \boldsymbol{\mu}^{(j)}\rangle \cdot \mu^{(j)}_h
  = 
  \sum_{t\in [r]} s^{(t)} (u^{(t)}_j-v^{(t)}_h)
  = 
  \sum_{t\in [r]} s^{(t)} u^{*(t)}_j
    \ .
\end{align*} 
Therefore
\begin{align}
  \frac{1}{\ell}\sum_{j\in[\ell]}\langle \boldsymbol{z}^{(A)}, \boldsymbol{\mu}^{(j)}\rangle \cdot \mu^{(j)}_h &=\frac{1}{\ell}\sum_{j\in[\ell]}
  \sum_{t\in [r]} s^{(t)} u^{*(t)}_j \label{ave:eq}\\
  &= \sum_{t\in [r]} s^{(t)} \frac{1}{\ell}\sum_{j\in[\ell]}  u^{*(t)}_j \nonumber\\
  &= \sum_{t\in [r]} s^{(t)} \overline{u^{*(t)}} = \overline{u^{(A)}}   \nonumber\ .
\end{align} 

Using Lemma~\ref{biasminest:lemma} and noting that our settings of $a,c$ (see Lemma~\ref{correctest:lemma}) have $c-a \gg a$ we obtain the claim:
 \begin{align*}
     \E_{\rho,\mathcal{A}}\left[\frac{1}{\ell}\sum_{j\in [\ell]} \langle \boldsymbol{z}^{(A)}, \boldsymbol{\mu}^{(j)}\rangle \cdot \mu^{(j)}_h\right] &\approx \frac{r}{\ell} \cdot \frac{2\sigma}{c+a}\\
      \Var_{\rho,\mathcal{A}}\left[\frac{1}{\ell}\sum_{j\in [\ell]} \langle \boldsymbol{z}^{(A)}, \boldsymbol{\mu}^{(j)}\rangle \cdot \mu^{(j)}_h\right] &\approx \frac{r}{\ell} \cdot \sigma^2
  \end{align*}
\end{proof}

\section{Proofs for \texttt{MeanEstAttack} Properties} \label{minestproofs:sec}

We provide the proof details for Lemma~\ref{biasminest:lemma}.
At a high level, we show that for each  $t\in [r]$, the value $s^{(t)}$ is correlated with the sign of the normalized bias of the sample:
\begin{equation} \label{normalizedbias:eq}
    \alpha^{(t)} := \frac{\overline{u^{*(t)}}}{\sigma} =  \frac{\overline{u^{(t)}}-v^{(t)}}{\sigma}\ .
\end{equation}

For $t\in [r]$ with reporting function $p$  we define (superscript $(t)$ omitted for readability from $p$, $\pi$)
\begin{align}
  \pi(v) &:= \Pr[s^{(t)} = 1 \mid v^{(t)} = v] = \E_{\boldsymbol{u}\sim \mathcal{N}_\ell(v,\sigma^2)} [p(\boldsymbol{u})] \label{pidef:eq}\\
  \pi(v\mid \alpha) &:= \Pr\left[s^{(t)}=1 \mid v^{(t)} = v \textbf{ and } \frac{\overline{u}-v}{\sigma}=\alpha \right]=\E_{\boldsymbol{u}\sim \mathcal{N}_\ell(v,\sigma^2) \mid \frac{\overline{u}-v}{\sigma}=\alpha} [p(\boldsymbol{u})] \label{picond:eq}
\end{align}
The probability of $s^{(t)}=1$
 when $v^{(t)} = v$ and
the probability of $s^{(t)}=1$ when $v^{(t)}=v$ {\em and } $\alpha^{(t)} = \alpha$.
Note that since $\overline{u} \sim \mathcal{N}(v,\sigma^2/\ell)$ we have $\alpha^{(t)} \sim \mathcal{N}(0,1/\ell)$.

 We will establish the following:
\begin{lemma} [Gap Lemma] \label{gap:coro}
For $0< \alpha \leq b\leq a$ (with probability $1-O(\delta)$ over the distribution of $\alpha$) and for any correct reporting function $p$ (and respective $\pi$)
\[
\E_{v\sim \mathcal{U}[(a ,(c+2b)\cdot\sigma ]} [\pi(v \mid \alpha) - \pi(v \mid -\alpha) ] \approx \frac{2\alpha}{c-a+2b} .
\]
\end{lemma} 

This Lemma implies that when
$v^{(t)} \sim \sigma\cdot [a,c+2b]$ then any correct estimator would exhibit a gap in the probability of reporting $s^{(t)}=1$, depending on the sign of the normalized bias $\alpha^{(t)}$. The gap increases with the magnitude $|\alpha^{(t)}|$.  Note that the statement is on the expected  {\em difference} between the reporting probabilities for positive or negative bias. The expected probabilities 
$\E_{v\sim \mathcal{U}[(a ,(c+2b)\cdot\sigma ]} [\pi(v \mid \alpha)$ and
$\E_{v\sim \mathcal{U}[(a ,(c+2b)\cdot\sigma ]} [\pi(v \mid -\alpha)$ 
can vary wildly, be both close to $0$ or close to $1$, for different correct reporting functions.

\subsection{Proof of Lemma~\ref{gap:coro}}

We will use two Lemmas.  The first uses the fact that the reporting probability as a function of 
$v^{(t)}$ has total increase of almost $1$ between $v=a\cdot \sigma$  to $v = c\cdot \sigma$ to establish its average increases over uniformly random sub-intervals of length $2\alpha\sigma$:
\begin{lemma} [Average increase property] \label{correctestavg:lemma}
For $0< \alpha \leq b \leq a$ (with probability $1-O(\delta)$ over the distribution)
\[
\E_{v\sim \mathcal{U}[(a ,(c+2b)\cdot\sigma ]} [\pi(v \mid \alpha) - \pi(v-2\cdot \alpha\cdot\sigma \mid \alpha) ] \approx \frac{2\alpha}{c-a+2b} .
\]
\end{lemma}
\begin{proof}
From Corollary~\ref{correctp:coro}, for almost all values of $\alpha$ (probability $1-O(\delta)$ over the distribution  of $\alpha\sim \mathcal{N}(0,1/\ell)$) we must have that
the function $\pi( v \mid \alpha)$ is very close (within a small constant of our choice) to $0$ when
$|v|\leq a\cdot \sigma$ and very close to $1$ when $|v| \geq c\cdot \sigma$.
\begin{align*}
& \E_{v\sim \mathcal{U}[a\cdot\sigma ,(c+2b)\cdot\sigma ]} [\pi(v  \mid \alpha) - \pi(v -2\alpha\cdot\sigma \mid \alpha) ]\\
&= \frac{1}{c-a+2b} \int_{a\cdot\sigma}^{(c+2b)\cdot\sigma } \pi(v \mid \alpha) -     \pi(v -2\alpha\cdot\sigma \mid \alpha) dv \\
&= \frac{1}{c-a+2b} \left(\int_{a\cdot\sigma}^{(c+2b )\cdot\sigma } \pi(v \mid \alpha) dv - \int_{(a-2\alpha)\cdot\sigma}^{(c+2b-2\alpha)\cdot\sigma } \pi(v \mid \alpha) dv \right)\\
&= \frac{1}{c-a+2b} \left(- \int_{(a-2\alpha)\cdot\sigma}^{a \cdot\sigma} \pi(v \mid \alpha) dv + \int_{(c+2b-2\alpha)\cdot\sigma}^{(c+2b)\cdot\sigma} \pi(v \mid \alpha) dv \right) \approx \frac{2\alpha}{c-a+2b}  \end{align*}
\end{proof}

\ignore{
\begin{lemma} [Average increase property] \label{correctestavg:lemma}
For any $0<\gamma, \alpha \leq a$ (with v.h.p over the distribution of $\alpha$)
\[
\E_{v\sim \mathcal{U}[(0 ,(c+a)\cdot\sigma ]} [\pi(v+\gamma\cdot \sigma \mid \alpha) - \pi(v-\gamma\cdot \sigma \mid \alpha) ] = \Theta(\gamma/c) .
\]
(in particular, this holds when not conditioning on $\alpha$).
\end{lemma}
\begin{proof}
From Lemma~\ref{correctest:lemma},
the function $\pi( v \mid \alpha)$ is very close to $0$ when
$|v|\in [0,a\cdot \sigma]$ and very close to $1$ when $v\geq c\cdot \sigma$.
Its mean growth over an sub-interval of length $2\gamma$ is therefore $\propto \frac{2\gamma}{c+a}$:
\begin{align*}
& \E_{v\sim \mathcal{U}[0,(c+a)\cdot\sigma ]} [\pi(v+\gamma\cdot \sigma \mid \alpha) - \pi(v-\gamma\cdot \sigma \mid \alpha) ]\\
&= \frac{1}{c+a} \int_{0}^{(c+a)\cdot\sigma } \pi(v +\gamma\cdot\sigma \mid \alpha) -     \pi(v -\gamma\cdot\sigma \mid \alpha) dv \\
&= \frac{1}{c+a} \left(\int_{(\gamma)\cdot\sigma}^{(c+a +\gamma)\cdot\sigma } \pi(v \mid \alpha) dv - \int_{(-\gamma)\cdot\sigma}^{(c+a-\gamma)\cdot\sigma } \pi(v \mid \alpha) dv \right)\\
&= \frac{1}{c+a} \left(- \int_{(-\gamma)\cdot\sigma}^{\gamma \cdot\sigma} \pi(v \mid \alpha) dv + \int_{(c+a-\gamma)\cdot\sigma}^{(c+a+\gamma)\cdot\sigma} \pi(v \mid \alpha) dv \right) \approx \frac{2\gamma}{c+a}  
\end{align*}
\end{proof}
}

The following Lemma shows that the probability of
reporting $s^{(t)}=1$, is equally likely when 
$v^{(t)} = v + \alpha\cdot\sigma$ and $\overline{u} = v$ and when $v^{(t)} = v - \alpha\cdot\sigma$ and $\overline{u} = v$:
\begin{lemma}\label{sym:lemma}
For any $\alpha>0$, reporting function $p()$, and $v$, 
$\pi(v +\alpha\cdot\sigma \mid -\alpha) = \pi(v -\alpha\cdot\sigma  \mid +\alpha)$.
\end{lemma}
\begin{proof}

From \eqref{picond:eq}:
\begin{align*}
\pi(v +\alpha\cdot\sigma \mid -\alpha) &= \frac{1}{\int_{\boldsymbol{u} \mid \overline{u}=v} f_{v +\alpha\cdot\sigma}(\boldsymbol{u}) d \boldsymbol{u} } \cdot 
\int_{\boldsymbol{u} \mid \overline{u}=v} f_{v +\alpha\cdot\sigma}(\boldsymbol{u}) \cdot p(\boldsymbol{u}) d \boldsymbol{u}\\
\pi(v -\alpha\cdot\sigma \mid \alpha) &= 
\frac{1}{\int_{\boldsymbol{u} \mid \overline{u}=v} f_{v +\alpha\cdot\sigma}(\boldsymbol{u}) d \boldsymbol{u} } \cdot 
\int_{\boldsymbol{u} \mid \overline{u}=v} f_{v -\alpha\cdot\sigma}(\boldsymbol{u}) \cdot p(\boldsymbol{u}) d \boldsymbol{u}
\end{align*}

The probability $\pi(v +\alpha\cdot\sigma \mid -\alpha)$ is the mean over $\boldsymbol{u}$ with $\overline{u}=v$ of $f_{v +\alpha\cdot\sigma}(\boldsymbol{u}) \cdot p(\boldsymbol{u})$.  Similarly, $\pi(v -\alpha\cdot\sigma \mid \alpha)$ is the mean over $\boldsymbol{u}$ with $\overline{u}=v$ of $f_{v -\alpha\cdot\sigma}(\boldsymbol{u}) \cdot p(\boldsymbol{u})$.
It therefore suffices to establish that
$f_{v -\alpha\cdot\sigma}(\boldsymbol{u})= f_{v +\alpha\cdot\sigma}(\boldsymbol{u})$ for $\boldsymbol{u}$ when $\overline{u}=v$ .

This is a 
property of a product of normal distributions: for any $\boldsymbol{u} \in \mathbb{R}^\ell$ and $\alpha$, the outcome
$\boldsymbol{u}$
is equally likely with  $\mathcal{N}_\ell(\overline{u}+\alpha\cdot\sigma,\sigma^2 )$ and $\mathcal{N}_\ell(\overline{u}-\alpha\cdot\sigma,\sigma^2)$:
\[
f_{\overline{u}+\alpha\cdot\sigma}(\boldsymbol{u}) =f_{\overline{u}-\alpha\cdot\sigma}(\boldsymbol{u}).
\]

We express $\boldsymbol{u}$ in terms of the empirical mean
$\overline{u} := \frac{1}{\ell} \sum_j u_j$
and
the "deviations" from the empirical mean
$(\Delta_j =  u_j - \overline{u})_{j=1}^{\ell}$.
Note that by definition $\sum_{j\in [\ell]} \Delta_j = 0$. 
Let $\alpha >0$,
using~\eqref{pdfnormal:eq} and $\sum_i \Delta_i=0$ we obtain:
\begin{align*}
    f_{\overline{u}+\alpha\cdot\sigma}(\boldsymbol{u}) 
    =& \prod_{i\in[\ell]} \frac{1}{\sigma\sqrt{2\pi}} \exp{-\left(\frac{\overline{u}+\Delta_i -(\overline{u}+\alpha\cdot\sigma)}{\sigma}\right)^2}\\
    =& \frac{1}{(\sigma\sqrt{2\pi})^\ell}\exp\left(\frac{-\sum_{i=1}^\ell (\Delta_i^2 + \alpha^2\cdot\sigma^2 -2 \alpha\cdot\sigma\Delta_i)}{\sigma^2} \right)\\
    =& \frac{1}{(\sigma\sqrt{2\pi})^\ell}\exp\left(\frac{-\ell \alpha^2\cdot\sigma^2 - \sum_{i=1}^\ell \Delta_i^2 +2 \alpha\cdot\sigma\cdot \sum_{i=1}^\ell \Delta_i)}{\sigma^2} \right)\\
    =& \frac{1}{(\sigma\sqrt{2\pi})^\ell}\exp\left(\frac{-\ell \alpha^2\cdot\sigma^2 - \sum_{i=1}^\ell \Delta_i^2}{\sigma^2} \right)=\frac{1}{(\sigma\sqrt{2\pi})^\ell}\exp\left(-\ell \alpha^2 - \frac{\sum_{i=1}^\ell \Delta_i^2}{\sigma^2} \right)
\end{align*}
The calculation for $f_{\overline{u}-\alpha\cdot\sigma}(\boldsymbol{u})$ is similar and  yields the same expression:
\begin{align*}
    f_{\overline{u}-\alpha\cdot\sigma}(\boldsymbol{u}) 
    =& \prod_{i\in[\ell]} \frac{1}{\sigma\sqrt{2\pi}} \exp{-\left(\frac{\overline{u}+\Delta_i -(\overline{u}-\alpha\cdot\sigma)}{\sigma}\right)^2}\\
    =& \frac{1}{(\sigma\sqrt{2\pi})^\ell}\exp\left(\frac{-\sum_{i=1}^\ell (\Delta_i^2 + \alpha^2\cdot\sigma^2 +2 \alpha\cdot\sigma\Delta_i)}{\sigma^2} \right)\\
    =& \frac{1}{(\sigma\sqrt{2\pi})^\ell}\exp\left(\frac{-\ell \alpha^2\cdot\sigma^2 - \sum_{i=1}^\ell \Delta_i^2 -2 \alpha\cdot\sigma\cdot \sum_{i=1}^\ell \Delta_i)}{\sigma^2} \right)\\
    =& \frac{1}{(\sigma\sqrt{2\pi})^\ell}\exp\left(-\ell \alpha^2 - \frac{\sum_{i=1}^\ell \Delta_i^2}{\sigma^2} \right)
\end{align*}

\end{proof}

Lemma~\ref{gap:coro} follows as a corollary of 
Lemma~\ref{correctestavg:lemma} and Lemma~\ref{sym:lemma}: From Lemma~\ref{correctestavg:lemma} we get that
$\pi(v-2\cdot \alpha\cdot\sigma \mid \alpha) =  \pi(v \mid -\alpha)$.  By making this substitution in the statement of Lemma~\ref{sym:lemma} we obtain Lemma~\ref{gap:coro}.

\subsection{Output properties} 

In this section we provide a proof of Lemma~\ref{biasminest:lemma} that states properties of the expectation and variance of the average entry value of the output vector. Additionally, we bound from below the $\ell_2$ norm (Lemma~\ref{squaredbound:lemma}).  We use the bound on the $\ell_2$ norm in Section~\ref{AMS:sec}.

\begin{proof} [Proof of Lemma~\ref{biasminest:lemma}]

From linearity of expectation,
\begin{equation} \label{claimE:eq}
\E_{\mathcal{A}}\left[\overline{u^{(A)}}\right] = 
   \E\left[\sum_{t\in [r]} s^{(t)}(\overline{u}^{(t)}-v^{(t)}))\right]=
\sigma \E\left[\sum_{t\in[r]} s^{(t)} \alpha^{(t)}\right] = \sigma \sum_{t\in [r]} \E[s^{(t)} \alpha^{(t)}] ,
\end{equation}
where $\alpha^{(t)}$ is the normalized bias \eqref{normalizedbias:eq}.
We now consider the contribution 
of a single query $t \in [r]$:
\begin{equation} \label{onet:eq}
   \E[s^{(t)}\alpha^{(t)}]\ .
\end{equation}

Recall that the distribution of  
$\alpha^{(t)}$
 is 
$\mathcal{N}(0,1/\ell)$.
We condition on
$|\alpha^{(t)}| = \alpha$ for $\alpha \geq 0$ which has the density function
\begin{equation} \label{evennormal:eq}
2\sqrt{\frac{\ell}{2\pi}}e^{-\frac{1}{2}\ell \alpha^2} .
\end{equation}

We have $s^{(t)}\alpha^{(t)}= \alpha$ if
$\alpha^{(t)} = \alpha$
and $s^{(t)}=1$ or when 
$\alpha^{(t)} = \alpha = -\alpha$
and $s^{(t)}=-1$.  Similarly,  $s^{(t)}\alpha^{(t)} = -\alpha$ if
$\alpha^{(t)} = \alpha$
and $s^{(t)}=-1$ or 
when 
$\alpha^{(t)} = -\alpha$
and $s^{(t)}=1$.

We now express the expected contribution under this conditioning, for fixed $v^{(t)}$ and fixed correct reporting function $p$.  We use that since the distribution of $\alpha^{(t)}$ (for fixed $v^{(t)}$) is symmetrical, there is equal probability $1/2$ for each of
$\alpha^{(t)} = \pm \alpha$. We use
$\pi(v\mid\alpha)$ as defined in
\eqref{picond:eq}.
\begin{align*}
\lefteqn{\E_{p}[s^{(t)}\alpha^{(t)} \mid \alpha^{(t)} =\alpha ] =\E_{\mathcal{A}}[s^{(t)}(\frac{\overline{u}^{(t)}-v^{(t)}}{\sigma}) \mid \frac{|\overline{u}^{(t)}-v^{(t)}|}{\sigma}=\alpha ] =}\\
&= \frac{1}{2} \alpha \left(\pi(v^{(t)} \mid \alpha)-\pi(v^{(t)} \mid -\alpha) -(1-\pi(v^{(t)} \mid \alpha)) + (1-\pi(v^{(t)}\mid -\alpha))\right) \\
&=  \frac{1}{2} \alpha \left(2\pi(v^{(t)}\mid\alpha) - 2\pi(v^{(t)} \mid -\alpha)\right) =  \alpha  \left(\pi(v^{(t)}\mid\alpha) - \pi(v^{(t)} \mid -\alpha)\right)
\end{align*}

We take the expectation over $v^{(t)} \sim \mathcal{U}[a\cdot\sigma,(c+2b)\cdot\sigma]$ and apply Lemma~\ref{gap:coro}. We obtain that for $\alpha\in [0,b]$:
\begin{align*}
\lefteqn{\E_{v^{(t)},p}[s^{(t)} \alpha^{(t)} \mid \alpha^{(t)} =\alpha ] =}\\
&= \alpha  \E_{v\sim \mathcal{U}[a\cdot\sigma,(c+2b)}[\pi(v\mid\alpha) - \pi(v \mid -\alpha)] \approx \alpha \frac{2\alpha}{c-a+2b} = \frac{2\alpha^2}{c-a+2b}  
\end{align*}

Finally, we now take the expectation over $\alpha \geq 0$ using the density
\eqref{evennormal:eq}:
\begin{align}
\E_{v^{(t)},p}[s^{(t)} \alpha^{(t)}] &=
\int_0^\infty (2\cdot\sqrt{\frac{\ell}{2\pi}}e^{-\frac{1}{2}\ell \alpha^2})
\E_{v,p}[s^{(t)}\alpha^{(t)} \mid |\alpha^{(t)}|=\alpha ] 
d\alpha \nonumber \\
&\approx  \frac{2}{c-a+2b} 2\cdot\sqrt{\frac{\ell}{2\pi}} \int_0^\infty \alpha^2 e^{-\frac{1}{2}\ell \alpha^2} d\alpha \nonumber \\
&=  \frac{4}{c-a+2b} \cdot\sqrt{\frac{\ell}{2\pi}} \ell^{-3/2} \sqrt{\frac{\pi}{2}} = \frac{1}{\ell} \cdot \frac{2}{c-a+2b} \label{pertE:eq}
  \end{align}
Noting that the approximation holds as the 
contribution of $\alpha\in [b,\infty)$ is negligible when the condition on $b$ holds: 
\begin{align*}
& \int_b^\infty (2\cdot\sqrt{\frac{\ell}{2\pi}}e^{-\frac{1}{2}\ell \alpha^2}) \left| \E_{v^{(t)},p}[s^{(t)} \alpha^{(t)}  \mid |\alpha^{(t)}| =\alpha ] \right | d\alpha  \\
&\leq \int_b^\infty  (2\cdot\sqrt{\frac{\ell}{2\pi}}e^{-\frac{1}{2}\ell \alpha^2}) \alpha d\alpha \leq  \sqrt{\frac{2}{\pi \ell}} e^{- \ell b^2/2} 
\end{align*}

Substituting \eqref{pertE:eq} in  \eqref{claimE:eq} we establish  claim \eqref{biasB:eq} of the Lemma:
$\E_{\mathcal{A}}[\overline{u^{(A)}}] \approx  \frac{r}{\ell} \sigma \cdot \frac{2}{c-a+2b}$.

We now bound (using independence of $s^{(t)}\alpha^{(t)}$ for different $t$ and thus linearity of the variance)
\begin{equation} \label{claimV:eq}
\Var_{\mathcal{A}}\left[\overline{u^{(A)}}\right] = 
   \Var\left[\sum_{t\in [r]} s^{(t)}(\overline{u}^{(t)}-v^{(t)}))\right]=
\sigma^2 \Var\left[\sum_{t\in[r]} s^{(t)} \alpha^{(t)}\right] = \sigma^2 \sum_{t\in [r]} \Var[s^{(t)} \alpha^{(t)}] .
\end{equation}

We bound
  \begin{equation} \label{vart:eq}
\Var[s^{(t)}\alpha^{(t)}] .
\end{equation}
We condition on
$|\alpha^{(t)}| = \alpha$ and take expectation over $v$, applying Lemma~\ref{gap:coro}:
We have $s^{(t)}\alpha^{(t)} = \alpha$ with probability
\[
\frac{1}{2}\E_v[\pi(v \mid \alpha) - (1-\pi(v \mid -\alpha))]= \frac{1}{2}+ \frac{\alpha}{c-a+2b}
\]
and (similarly)
$s^{(t)}\alpha^{(t)} = - \alpha$ with probability $\frac{1}{2}- \frac{\alpha}{c-a+2b}$. 
We obtain
\begin{align}
    \lefteqn{\Var[s^{(t)}\alpha^{(t)} \mid |\alpha^{(t)} | = \alpha] =}\nonumber\\
    &= \E[(s^{(t)}\alpha^{(t)})^2\mid |\alpha^{(t)} | = \alpha  ] -
    \E[s^{(t)}\alpha^{(t)}\mid |\alpha^{(t)} | = \alpha ]^2 = \alpha^2 \left(1 - \left(\frac{2\alpha}{c-a+2b}\right)^2 \right)\ .\label{varalpha:eq}
\end{align} 

Since $\alpha \in [0,b]$ 
we have
\[
1 - \frac{2b}{c-a+2b} \leq 1 - \frac{2\alpha}{c-a+2b} \leq 1\ .
\]
Take the expectation of \eqref{varalpha:eq} over $\alpha \geq 0$  using these bounds and density
\eqref{evennormal:eq} we obtain
\begin{align}
    \frac{1}{\ell} \left(1 - \frac{2b}{c-a+2b}  \right)^2 \leq \Var_{v^{(t)},p}[s^{(t)}\alpha^{(t)}] \leq \frac{1}{\ell}\ .\label{bound:eq}
\end{align}
Claim \eqref{biasvarB:eq} in the statement of the Lemma follows by substituting \eqref{bound:eq}  in \eqref{claimV:eq}.
In the regime where 
$c-a \gg b$, $\frac{2b}{c-a+2b}$ is small and we have
\begin{align*}
\Var_{v^{(t)},p}[s^{(t)}\alpha^{(t)}] \approx \frac{1}{\ell}\ .
\end{align*}

 \end{proof}
 
We will use the following in our analysis of the attack on the AMS sketch: 
 \begin{lemma} \label{squaredbound:lemma}
 Under the conditions of Lemma~\ref{biasminest:lemma}, the following holds:
\begin{align}
     \E_{\mathcal{A}}\left[\frac{1}{\ell}\|\boldsymbol{u}^{(A)}\|^2_2 \right]  
     &\geq  
r\sigma^2\left(1-O(\frac{1}{\sqrt{\ell}}) + \Omega\left(\frac{r}{\ell^2} \left(\frac{2}{c-a+2b}\right)^2 \right) \right) \label{biasSqB:eq}
\end{align}      
  \end{lemma}     
\begin{proof}
We now establish Claim~\eqref{biasSqB:eq}
\begin{align}
\E_{\mathcal{A}}\left[\frac{1}{\ell}\|\boldsymbol{u}^{(A)}\|^2_2 \right] &= \frac{1}{\ell} \sum_{j\in [\ell]}  \E_{\mathcal{A}}\left[(u^{(A)}_j)^2\right] \nonumber\\
&= \frac{1}{\ell} \sum_{j\in [\ell]} \left( \Var_{\mathcal{A}} \left[ u^{(A)}_j \right] + \E_{\mathcal{A}}\left[ u^{(A)}_j \right]^2 \right)\nonumber\\
&= \frac{1}{\ell} \sum_{j\in [\ell]}  \Var_{\mathcal{A}} \left[ u^{(A)}_j \right] + \frac{1}{\ell} \sum_{j\in [\ell]} \E_{\mathcal{A}}\left[ u^{(A)}_j \right]^2 \label{sqsubbound:eq} \ .
\end{align}

We apply Claim~\eqref{biasB:eq} to bound the second term:
\begin{align}
 \frac{1}{\ell} \sum_{j\in [\ell]}  \E_{\mathcal{A}}\left[ u^{(A)}_j \right]^2 &\geq \E_{\mathcal{A}}\left[\frac{1}{\ell} \sum_{j\in [\ell]} u^{(A)}_j \right] & \, \text{applying the inequality $E[X^2]\geq E[X]^2$}\nonumber \\
 &= \E_{\mathcal{A}}[\overline{u^{(A)}}]^2 \approx  \left(\frac{r}{\ell} \sigma^2 \cdot \frac{2}{c-a+2b}\right)^2 & \, \text{Claim~\eqref{biasB:eq}} \label{Esqbound:eq}
\end{align}

We now bound from below 
\begin{align}
    \sum_{j\in [\ell]} \Var_{\mathcal{A}} \left[ u^{(A)}_j \right] &=
    \sum_{j\in [\ell]} \Var_{\mathcal{A}}\left[ \sum_{t\in[r]}  s^{(t)} u^{*(t)}_j \right] \geq r\cdot\ell\cdot\sigma^2 \left(1 -O(\frac{1}{\sqrt{\ell}}) \right)\ . \label{Vsqbound:eq}
\end{align}
We recall that $u^{(t)}_j$ are i.i.d.\ $\mathcal{N}(0,\sigma)$.
We consider a bound where $s^{(t)}\in\{\pm 1\}$ is selected in a worst case way (as to minimize variance) according to the actual $\boldsymbol{u}^{*(t)}\sim \mathcal{N}_{\ell}(0,\sigma)$.
With $\ell=1$ the variance is minimized when the correlation are maximized which is done by having $s^{(t)}u^{*{t)}}$ all with the same sign, that is, behaving like i.i.d.\ $|\mathcal{N}(0,\sigma)|$ that have variance $(1- 2/\pi)\sigma^2$ per term with respective variance 
$r(1-2/\pi)\sigma^2$ for the sum. Since the same $s^{(t)}$ multiplies all entries in $\boldsymbol{u}^{*(t)}$, we can in expectation have at most $1/2+ \sqrt{\frac{2}{\pi \ell}}$ of these values with the same sign. Over the total of $r\cdot \ell$ terms we get
expected contribution of at least
$ 2 \sqrt{\frac{2}{\pi \ell}}\cdot (1- 2/\pi)\sigma^2 + (1- 2 \sqrt{\frac{2}{\pi \ell}})\cdot \sigma^2$.

Finally, we substitute \eqref{Esqbound:eq} and \eqref{Vsqbound:eq} in \eqref{sqsubbound:eq} and obtain the claim
\begin{align*}
    \frac{1}{\ell} \sum_{j\in [\ell]}  \E_{\mathcal{A}}\left[ u^{(A)}_j \right]^2 &\geq r\cdot\sigma^2 \left(1 -O(\frac{1}{\sqrt{\ell}}) \right) + \left(\frac{r}{\ell} \sigma^2 \cdot \frac{2}{c-a+2b}\right)^2 \\
    &= r\sigma^2 \left(1 - O(\frac{1}{\sqrt{\ell}}) + \frac{r}{\ell^2} (\frac{2}{c-a+2b})^2 \right)
\end{align*}

\end{proof}

\section{Extension: Universal Attack on the AMS Sketch} \label{AMS:sec}

\begin{algorithm2e}
{\small
    \caption{\texttt{Universal Attack on AMS norm estimation}}
    \label{algo:AMSattack}
    \DontPrintSemicolon
    \KwIn{$\tau$, $\epsilon$, $\delta$ \tcp*{Estimator parameters (Definition~\ref{correctnormest:def})}
    Initialized AMS $\sketch_\rho$ with parameters $(n,\ell)$, number of queries $r$, tail support size $m$}
    $\sigma \gets \tau \sqrt{1/2}$; $c \gets \sqrt{1+2\epsilon} $; $a \gets 1$;     $b\gets \Theta(\frac{1}{\sqrt{\ell}} \ln(\sqrt\frac{\ell}{\epsilon})$\;
    
    $(h,(\boldsymbol{z}^{(t)})_{t\in [r]}) \gets$
    \texttt{AttackTails}$(n,m,r)$\tcp*{Choose tails (Algorithm~\ref{algo:tails})}
    $(\boldsymbol{z}^{(t)} \gets \boldsymbol{z}^{(t)} \cdot \frac{\sigma}{\sqrt{m}})_{t\in [r]} $\tcp*{Rescale tails to have $\ell_2$ norm $\sigma$}
 \For(\tcp*[h]{Generate Query Vectors}){$t\in [r]$}{
    $v^{(t)}_h\sim \mathcal{U}[a\cdot \sigma, (c+2b)\cdot \sigma]$\;
    $\boldsymbol{v}^{(t)} \gets  v^{(t)}_h \boldsymbol{e}_h +
    \boldsymbol{z}^{(t)}$\tcp*{Query vectors}
    }
    \For(\tcp*[h]{Apply Query Response}){$t\in [r]$}
    {
      Choose an $(\epsilon,\delta,\tau)$-correct estimator $M^{(t)}$ \tcp*{Definition~\ref{correctnormest:def}, may depend on $(v^{(t')}_h,s^{(t')},M^{(t')})_{t'<t}$ and $\rho$ }
       $s^{(t)} \gets M^{(t)}(\sketch_\rho(\boldsymbol{v}^{(t)})$\tcp*{Apply estimator to sketch}
      }
     \Return{$\boldsymbol{z}^{(A)} \gets \sum_{t\in[r]} s^{(t)} \boldsymbol{z}^{(t)} $}\tcp*{Adversarial input}
    }
  \end{algorithm2e}

The AMS sketch~\cite{ams99} and the related Johnson Lindenstrauss transform~\cite{JLlemma:1984} are linear maps of input vectors $\boldsymbol{v} \in \mathbb{R}^n$ to their 
sketches in $\mathbb{R}^\ell$. The sketches support  recovery of $\ell_2$ norms and distances.
In this section we describe an attack on the AMS sketch that applies with any correct norm estimator.  

The sketch is specified by parameters $(n,\ell)$, 
where $n$ is the dimension of input vectors and $\ell$ is the number of measurements.
The internal randomness $\rho$ specifies a set of random hash functions 
$s_{j}:[n]\rightarrow \{-1,1\}$ ($ j\in [\ell]$) so that $\Pr[s_{j}(i)=1]=1/2$.  
These hash functions define $\ell$ measurement vectors
$\boldsymbol{\mu}^{(j)}$ ($j\in [\ell]$):
\[ \mu^{(j)}_i := s_{j}(i) .\]
The AMS sketch analysis holds with 4-wise independent functions but our attack analysis holds with full randomness.

Norm estimators $M$ applied to the sketch guarantee that
when $\ell = O(\epsilon^{-2} \log(1/\delta))$
\[
\Pr_\rho[| \|\boldsymbol{v}\|_2^2 -  M(\sketch_\rho(\boldsymbol{v}))^2 |   \geq \epsilon \|\boldsymbol{v}\|_2^2] < \delta\ .
\]
In particular, for any
$\boldsymbol{v}\in R^n$ and $j\in[\ell]$ it holds that,
$\E_{\rho}[\langle \boldsymbol{\mu},\boldsymbol{v}\rangle ^2] = \|\boldsymbol{v}\|^2_2$ and
$\Var_{\rho}[\langle \boldsymbol{\mu},\boldsymbol{v}\rangle ^2] = O(\|\boldsymbol{v}\|^2_2)$.

\begin{definition} [Adversarial input for AMS] \label{adversarialinput:def}
 We say that an attack $\mathcal{A}$ that is applied to a sketch with randomness $\rho$ and outputs
 $\boldsymbol{z}^{(A)}\in \{-1,0,1\}^n$ is   
  $(\xi,\beta)$-{\em adversarial} (for $\xi>0$) if with probability at least $1-\beta$ over the randomness of $\rho,\mathcal{A}$ it holds that:
  \begin{equation}
\Pr_{\rho,\mathcal{A}}\left[\frac{1}{\ell} \sum_{j\in [\ell]}  \langle \boldsymbol{\mu}^{(j)}, \boldsymbol{z}^{(A)} \rangle^2
\geq (1+\xi) \cdot \| \boldsymbol{z}^{(A)}\|^2_2\right] \geq 1-\beta\ .
  \end{equation}
\end{definition}
Note that when $\ell \gg \frac{1}{ \epsilon^2}\ln(1/\delta)$, for any $\boldsymbol{v}$, we have
\begin{equation}
\Pr_{\rho\sim\mathcal{D}} \left[ \frac{1}{\ell} \sum_{j\in [\ell]}  \langle \boldsymbol{\mu}^{(j)}(\rho), \boldsymbol{v} \rangle^2
\geq (1+\epsilon) \cdot \| \boldsymbol{v}\|^2_2 \right] \leq \delta\ .
\end{equation}

Algorithm~\ref{algo:AMSattack} describes an
attack that constructs an adversarial input for a sketch with randomness $\rho$.  Here we consider both the accuracy
$\epsilon$ (relative error) of the estimator, which is at least $\epsilon \geq 1/\sqrt{\ell}$ (otherwise correct estimators do not exist) and the bias $\xi$ in the product of the attack.  The bias also must have $\xi = \Omega(1/\sqrt{\ell})$ (otherwise $\xi$ is within the variability of the output) and naturally we may want $\xi \geq \epsilon$, that is, bias is larger than the accuracy of the estimators we attack.  

We establish the following (the proof is presented after a brief discussion that interprets the result):
\begin{lemma}. \label{adverseAMS:lemma}
For $\epsilon, \xi = \Omega(1/\sqrt{\ell})$, any constant $\beta>0$,
and attack size $r = O(\xi \ell^2 \epsilon^2)$,
the output of
Algorithm~\ref{algo:AMSattack} is $(\xi,\beta)$-adversarial.
\end{lemma}

Intuitively, more accurate estimators (that is, smaller $\epsilon$) leak more information on the randomness and hence are easier to attack and we expect attack size to increase with $\epsilon$. A larger bias $\xi$ is harder to accrue and would require a larger attack and hence attack size also increases with $\xi$.  

Interestingly, for the AMS sketch, the sketch size parameter $\ell$ controls both 
robustness and accuracy, where with less accurate responses we can support more queries.  This is in contrast to heavy hitters via $\countsketch$, where the parameter $b$ controls the accuracy and we use $\ell$ to control confidence and robustness. (In principle, a larger $\ell$ in $\countsketch$ can also be used to increase accuracy but much less effectively for a given sketch size than increasing $b$). We therefore study the attack size on the AMS sketch in terms of both $\epsilon$ and $\xi$.

The Lemma implies that to generate an $O(1)$-adversarial input for a sketch of size $\ell$ when estimators have accuracy $\epsilon = \Theta(1)$ (what we can get in the non-adaptive setting from a sketch of size $\ell=O(\log(n/\delta))$) it suffices to use $O(\ell^2)$ queries. In this case there is a matching upper bound (up to logarithmic factors) using the "wrapper" method of~\cite{HassidimKMMS20}. 
 The required attack size when $\xi = O(\epsilon)$ is $r= O(\ell^2 \epsilon^3)$. 
 For the smallest applicable accuracy $\epsilon = O(1/\sqrt{\ell})$ we obtain that $r= O(\sqrt{\ell})$ queries suffice. In this case, there is a gap with the best known upper bound which is the trivial bound $r=\Omega(1)$ of answering a single query with the sketch. 
 This gap leaves an interesting open question. 
 Our attack has a simple structure where the only adaptive input is the final adversarial one.  We conjecture that there is more robustness in this case (thus a better upper bound) whereas a more complex attack with multiple adaptive rounds may yield a tight lower bound of $r=\tilde{O}(1)$ for the general case.

\subsection{Proof of Lemma~\ref{adverseAMS:lemma}} 
The algorithm and analysis of the attack on the AMS sketch are similar, but simpler, than the attack on $\countsketch{}$.
We use query vectors of a similar form with varying $v^{(t)}_h$ and fixed tail norm scaled to $\sigma$.
The distribution of
the measurements of the tails (multiplied by $\mu^{(j)}_h$)  $u^{*(t)}_j$ is $m^{-1/2} \mathcal{B}(m,1/2)-m/2$, which approaches
$\mathcal{N}(0,\sigma^2)$ for "large enough" $m$.  These values are i.i.d.\ for $j\in \ell$ and $t\in [r]$. 

The sketch content that is "visible" to the estimator has respective 
values $\boldsymbol{u}^{(t)} := v_h^{(t)}\boldsymbol{1}_\ell + \boldsymbol{u}^{*(t)}$
with respective distribution
$\mathcal{N}_\ell(v_h^{(t)},\sigma^2)$.
In our attack we choose
$v^{(t)}_h \sim \mathcal{U}[a,c+2b]\cdot \sigma$, where $a,b,c$ are specified below.
We assume that at each step $t\in [r]$, the query response algorithm chooses an estimator that is correct with respect to $(\epsilon,\delta)$ and $\tau$. The estimator returns only one bit, essentially comparing the norm to $\tau$ with accuracy $\epsilon$ and confidence $1-\delta$:
\begin{definition} [Correct $\ell_2$-norm estimator]\label{correctnormest:def}
A norm estimator is correct for parameters $(\epsilon,\delta)$ and $\tau$
if
$\|\boldsymbol{v}\|^2_2 \geq (1+\epsilon)\tau^2$, the output is $1$ with probability $\geq 1-\delta$ and if 
$\|\boldsymbol{v}\|^2_2 \leq \tau^2$ then the output is $-1$ with probability $\geq 1-\delta$.
\end{definition}

\begin{lemma}
A norm estimator applied to sketches of our query inputs is correct for parameters $(\epsilon,\delta)$ and $\tau$ if and only if 
its reporting function (see Definition~\ref{correctmean:def}) is correct for parameters
$(\delta,a,c,\ell,\sigma)$ with
\begin{align*}
  a &= \sqrt{(\tau/\sigma)^2 -1} \\
  c &= \sqrt{(1+\epsilon)(\tau/\sigma)^2 -1}\ .
\end{align*}
\end{lemma}
\begin{proof}
We have
\begin{align*}
    \|\boldsymbol{v}^{(t)}\|_2 \geq \tau & \iff
\sqrt{(v^{(t)}_h)^2 +\sigma^2} \geq \tau \implies
((v^{(t)}_h)^2 \geq \tau^2-\sigma^2 = \sigma^2((\tau/\sigma)^2-1) \\ & \implies |v^{(t)}_h|\geq \sigma\sqrt{((\tau/\sigma)^2-1)} =  a\cdot \sigma\ .\\
\|\boldsymbol{v}^{(t)}\|^2_2 \geq (1+\epsilon)\tau^2 &\iff |v^{(t)}_h|\geq c \cdot \sigma\ .
\end{align*}
\end{proof}
We analyze 
Algorithm~\ref{algo:AMSattack} using a correspondence to  Algorithm~\ref{algo:meanestattack}.  
We set the query vectors parameters according to the estimator parameters: We choose $\sigma^2 = \tau^2/2$ and thus
\begin{align*}
a &= 1 \\
c &= \sqrt{1+2\epsilon}
\end{align*}
Note that $c-a = \Theta(\epsilon)$ for $\epsilon<1$ and
$c-a = 1+\Theta(\sqrt{\epsilon})$ for $\epsilon > 1$.
We choose 
$b = \Theta(\frac{1}{\sqrt{\ell}} \ln(\sqrt\frac{\ell}{\epsilon})$ which satisfies the condition in Lemma~\ref{biasminest:lemma}. 
Since $\epsilon \geq \sqrt{\ln(1/\delta)/\ell}$ (otherwise correct estimators don't exist), we have $b = O(\epsilon)$ and thus
$c-a+2b = \Theta(\epsilon)$.

The adversarial input has
$\| \boldsymbol{z}^{(A)} \|^2_2 = r \cdot \sigma^2$ and to assess the bias we need to consider $\E_{\mathcal{A}}\left[\frac{1}{\ell}\|\boldsymbol{u}^{(A)}\|^2_2 \right]$. 
From Lemma~\ref{biasminest:lemma} using a direct argument
\eqref{Esqbound:eq} we obtain 
\[\E_{\mathcal{A}}\left[\frac{1}{\ell}\|\boldsymbol{u}^{(A)}\|^2_2 \right]  
     =  r\sigma^2 \frac{r}{\ell^2} \Omega\left(\left(\frac{2}{c-a+2b}\right)^2 \right) =
 \| \boldsymbol{z}^{(A)} \|^2_2 \frac{r}{\ell^2} \Omega(\epsilon^{-2}) .\]
 This establishes the claim of the Lemma for the 
regime of high bias $\xi = \Omega(1)$.

In order to derive tighter bounds for the  regime $\xi <1$ 
 we use a refined bound 
by applying Lemma~\ref{squaredbound:lemma}:
\begin{align*}
\E_{\mathcal{A}}\left[\frac{1}{\ell}\|\boldsymbol{u}^{(A)}\|^2_2 \right]  
     &=  
r\sigma^2\left(1-O(\frac{1}{\sqrt{\ell}}) + \frac{r}{\ell^2} \Omega\left(\left(\frac{2}{c-a+2b}\right)^2 \right) \right) \\
&= \| \boldsymbol{z}^{(A)} \|^2_2 \left(1-O(\frac{1}{\sqrt{\ell}})+\frac{r}{\ell^2} \Omega(\epsilon^{-2}) \right)
\end{align*}

Here for $\xi = \Omega(1/\sqrt{\ell})$,  we solve for $\xi = \Omega(r\epsilon^{-2}/\ell^2)$ and obtain the bound
$r = O(\xi \ell^2 \epsilon^2)$.

The confidence bounds of any fixed $\beta$ follows from the bound on the expectation using Markov inequality.

\section*{Conclusion}

Our results suggest interesting directions for followup work. We suspect that our attack technique can be generalized so that it applies with any randomized linear sketch that supports recovery of heavy hitters keys or even other families of features of the input vectors.  Our attack techniques might be applicable to trained neural networks as blackbox attacks that construct "random noise" that mimics the presence of a signal. As discussed in the introduction, such attacks complement \cite{HardtW:STOC2013,JayaramWZ:ICLR2020} that constructed "invisible" noise that obfuscates the input but not the output.

\newpage
\bibliographystyle{plain}  
\bibliography{references,robustHH}
\end{document}